\documentclass[letterpaper,11pt]{article}
\usepackage[margin=1in]{geometry}
\usepackage[utf8]{inputenc}
\usepackage[english]{babel}
\usepackage[T1]{fontenc}
\usepackage{amsmath,amssymb,amsthm}

\usepackage{hyperref}
\usepackage{xspace}
\usepackage{refstyle}
\usepackage[dvipsnames]{xcolor}

\usepackage{graphicx}

\newcommand{\gossip}{{\small\textsc{GOSSIP}}\xspace}
\newcommand{\pull}{{\small\textsc{PULL}}\xspace}
\newcommand{\push}{{\small\textsc{PUSH}}\xspace}

\newcommand{\bigO}{O}

\newcommand{\local}{{\small\textsc{LOCAL}}\xspace}
\newcommand{\congest}{{\small\textsc{CONGEST}}\xspace}

\newcommand{\bias}{b}

\newcommand{\general}{general\xspace}

\newcommand{\msg}{ {\texttt{msg}}}
\newcommand{\brotau}{{\tau}^{\mathrm{B}}}

\newcommand{\bw}{\mathbf{w}}

\newcommand{\bx}{\mathbf{x}}
\newcommand{\by}{\mathbf{y}}
\newcommand{\bX}{\mathbf{X}}
\newcommand{\bY}{\mathbf{Y}}
\newcommand{\sP}{\mathcal{P}}
\newcommand{\sS}{\mathcal{S}}
\newcommand{\sM}{\mathcal{M}}

\renewcommand{\ell}{k}

\newcommand{\vinput}[1]{\bx_{#1}}

\makeatletter
\def\tagform@#1{\maketag@@@{(\ignorespaces#1\unskip\@@italiccorr)}}
\makeatother

\newcommand{\twopart}[3]{$(#1,#2,#3)$-Two-Party Protocol\xspace}

\newcommand{\Prob}[1]{\mathbf{P}\left( #1 \right)}
\newcommand{\Probi}[1]{\mathbf{P}( #1 )}
\newcommand{\Ex}[1]{\mathbf{E}\left[ #1 \right]}
\def\E{{\mathbb E}}
\newcommand{\EEx}[2]{\mathbf{E}_{#1} \left[ #2 \right]}

\newcommand{\Var}[1]{\mathbf{Var}\left( #1 \right)}

\newcommand{\cond}{\, \vert \, }

\newcommand{\M}{\mathcal T}

\newcommand{\Confg}{\bf C}
\newcommand{\confg}{\bf c}
\newcommand{\conf}[1]{\bf c_{#1}}

\renewcommand{\le}{\leqslant}
\renewcommand{\leq}{\leqslant}
\renewcommand{\ge}{\geqslant}
\renewcommand{\geq}{\geqslant}
\renewcommand{\epsilon}{\varepsilon}

\newtheorem{theorem}{Theorem}[section]
\newtheorem{definition}[theorem]{Definition}

\newtheorem{lemma}[theorem]{Lemma}
\newtheorem{fact}[theorem]{Fact}
\newtheorem{corollary}[theorem]{Corollary}
\newtheorem{remark}[theorem]{Remark}

\newref{lem}{name=Lemma~}

\title{Consensus Needs Broadcast in Noiseless Models but can be Exponentially Easier in the Presence of Noise}

\author{Andrea Clementi \\ 
    {\footnotesize{}Università di Roma Tor Vergata}\\
    {\footnotesize{} Rome, Italy}\\
    {\footnotesize{}\texttt{clementi@mat.uniroma2.it}} 
    \and Luciano Gual\`a \\
    {\footnotesize{}Università di Roma Tor Vergata}\\
    {\footnotesize{} Rome, Italy}\\
    {\footnotesize{}\texttt{guala@mat.uniroma2.it}} 
    \and Emanuele Natale\\
    {\footnotesize{}Max Planck Institute for Informatics}\\
    {\footnotesize{}Saarbr\"ucken, Germany }\\
    {\footnotesize{}\texttt{emanuele.natale@mpi-inf.mpg.de}} 
    \and Francesco Pasquale \\
    {\footnotesize{}Università di Roma Tor Vergata}\\
    {\footnotesize{} Rome, Italy}\\
    {\footnotesize{}\texttt{pasquale@mat.uniroma2.it}} 
    \and Giacomo Scornavacca \\
    {\footnotesize{} Università degli Studi dell'Aquila}\\
    {\footnotesize{} L'aquila, Italy}\\
    {\footnotesize{}\texttt{giacomo.scornavacca@graduate.univaq.it}} 
    \and Luca Trevisan \\
    {\footnotesize{}U.C. Berkeley}\\
    {\footnotesize{} Berkeley, CA, United States}\\
    {\footnotesize{}\texttt{luca@berkeley.edu}} 
}
\date{}

\begin{document}
\maketitle

\begin{abstract}
Consensus and Broadcast are two fundamental problems in distributed computing, whose solutions have several applications. Intuitively, Consensus should be no harder than Broadcast, and this can be rigorously established in several models.
Can Consensus be {\em easier} than Broadcast?

In models that allow noiseless communication, we prove a reduction of (a suitable variant of) Broadcast to binary Consensus, that preserves the communication model and all complexity parameters such as randomness, number of rounds, communication per round, etc., while there is a loss in the success probability of the protocol. Using this reduction, we get, among other applications, the first logarithmic lower bound on the number of rounds needed to achieve Consensus in the uniform \gossip  model on the complete graph. The lower bound is tight and, in this model, Consensus and Broadcast are equivalent. 

We then turn to distributed models with noisy communication channels that have been studied in the context of  some bio-inspired   systems. In such models, only one noisy bit is exchanged when a communication channel is established between two nodes, and so one cannot easily simulate a noiseless protocol by using error-correcting codes. An $\Omega(\epsilon^{-2} n)$ lower bound on the number of rounds needed for Broadcast is proved by Boczkowski et al. [PLOS Comp. Bio. 2018] in one such model (noisy uniform \pull, where $\epsilon$ is a parameter that measures the amount of noise).  
In such model, we prove a new $\Theta(\epsilon^{-2} n \log n)$ bound for Broadcast and a $\Theta(\epsilon^{-2} \log n)$ bound for binary Consensus, thus establishing an exponential gap between the number of rounds necessary for Consensus versus Broadcast.
\end{abstract}

\thispagestyle{empty}
\clearpage
\setcounter{page}{1} 

\section{Introduction}
\label{sec:OLD_INTRO}

In this paper we investigate  the relation between   Consensus and   Broadcast, which are two of the most fundamental  algorithmic problems in distributed computing 
\cite{Dijkstra74,Dolev00,PSL80,R83}, and we study how the presence or absence of communication noise affects their complexity. 
 
In the (Single-Source) {\em Broadcast} problem, one node in a  network has an initial message $\msg$ and the goal is for all the nodes in the network to receive a copy of $\msg$.

In the  \emph{Consensus} problem, each 
of the $n$ nodes of a network starts with an input value  (which we will also call an \emph{opinion}), and the goal is for all the nodes to converge to a configuration in which they all have the same opinion (this is the {\em agreement} requirement) and this shared opinion is one held by at least one node at the beginning (this is the {\em validity} requirement). In the {\em Binary} Consensus problem, there are only two possible opinions, which we denote by 0 and 1. 

In the (binary) {\em Majority Consensus} problem \cite{AAE07,DGMSS11,PVV09} we are given the promise that one of the two possible opinions is initially held by at least $n/2 + b(n)$ nodes, where $b(n)$ is a parameter of the problem, and the goal is for the nodes to converge to a configuration in which they all have the opinion that, at the beginning, was held by the majority of nodes. Note that Consensus and Majority Consensus are incomparable problems: a protocol may solve one problem without solving the other.\footnote{A Consensus protocol is allowed to converge to an agreement to an opinion that was initially in the minority (provided that it was held by at least one node), while a Majority Consensus protocol must converge to the initial majority whenever the minority opinion is held by fewer than $n/2 -b$ nodes. On the other hand, a Majority Consensus problem is allowed to converge to a configuration with no agreement if the initial opinion vector does not satisfy the promise, while a Consensus protocol must converge to an agreement regardless of the initial opinion vector.}

Motivations for studying the Broadcast problem are self-evident. Consensus
and Majority Consensus are simplified models for the way inconsistencies and disagreements are resolved in social networks, biological models and peer-to-peer systems \cite{Doty14,HouseHunt,MosselNT14}.\footnote{The Consensus problem is   often studied in models in which nodes are subject to malicious faults, and, in that case, one has  motivations  from network security. In this paper we concentrate on models in which all nodes honestly follow the prescribed protocol and the only possibly faulty devices are the communication channels.}

In distributed model that severely restrict the way in which nodes communicate (to model constraints that arise in peer-to-peer systems or in social or biological networks), upper and lower bounds for the Broadcast problem give insights on the effect of the communication constraints on the way in which information can spread in the network. The analysis of algorithms for Consensus often give insights on how to break symmetry in distributed networks, when looking at how the protocol handles an initial opinion vector in which exactly half the nodes have one opinion and half have the other. The analysis of algorithms for Majority Consensus usually hinge on studying the rate at which the number of nodes holding the minority opinion shrinks. 

If the nodes are labeled by $\{ 1,\ldots, n\}$, and each node knows its label, then there is an easy reduction of binary Consensus to Broadcast: node $1$ broadcasts its initial opinion to all other nodes, and then all nodes agree on that opinion as the consensus opinion. Even if the nodes do not have known identities, they can first run a {\em leader election} protocol, and then proceed as above with the leader broadcasting its initial opinion.  Even in models where leader election is not trivial,  the best known Consensus protocol has, in all the cases that we are aware of,  at most the ``complexity'' (whether it's measured in memory per node, communication per round, number of rounds, etc.) of the best known broadcast protocol.

The question that we address in this paper is whether the converse hold, that is, are there  ways of obtaining a Broadcast protocol from a Consensus problem or are there gaps, in certain models, between the complexity of the two problems?

Roughly speaking, we will show that, in the presence of noiseless communication channels, every Consensus protocol can be used to realize a weak form of Broadcast. Since, in many cases, known lower bounds for Broadcast apply also to such weak form, we get  new lower bounds for Consensus. In a previously studied, and well motivated, distributed model with noisy communication, however, we establish an exponential gap between Consensus and Broadcast.

%Although some of our results hold in greater generality, we will state %them in reference to the uniform \gossip\ model in the complete graph, %as defined below.

\subsection{Communication and computational models} \label{ssec:Models}

In order to state and discuss our results we first introduce some distributed models and their associated complexity measures.

We study protocols defined on a communication network, described by an undirected graph $G=(V,E)$ where $V$ is the set of nodes, each one running an instance of the distributed algorithm, and $E$ is the set of pairs of nodes between which there is a communication link that allows them to exchange data. When not specified, $G$ is assumed to be the complete graph.

In {\em synchronous parallel} models, there is a global clock and, at each time step, nodes are allowed to communicate using their links.

In the \local\ model, there is no restriction on how many neighbors a node can talk to at each step, and no restriction on the number of bits transmitted at each step. There is also no restriction on the amount of memory and computational ability of each node. The only complexity measures is the number of rounds of communication. For example, it is easy to see that the complexity of Broadcast is the diameter of the graph $G$. The \congest\ model is like the \local\ model but the amount of data that each node can send at each time step is limited, usually to $O(\log n)$ bits. 

In the (general) \gossip\  model \cite{DGHILSSST87,H15}, at each time step, each node $v$
chooses one of its neighbors $c_v$
and {\em activates} the communication link $(v,c_v)$, over which communication becomes possible during that time step, allowing $v$ to send a message to $c_v$ and, simultaneously, $c_v$ to send a message to $v$. We will call $v$ the {\em caller of} $c_v$. In the \push\ variant, each node $v$ sends a message to its chosen neighbor $c_v$; in the \pull\ variant, each node $v$ sends a message to its callers. Note that, although  each node chooses only one neighbor, some nodes may be chosen by several others, and so they may receive several messages in the \push\ setting, or send a message to several recipients in the \pull\ setting. In our algorithmic results for the \gossip\ model, we will assume that each message exchanged in each time step is only one bit, and our negative results for the noiseless setting will apply to the case of messages of unbounded length. In the  {\em uniform} \gossip\ (respectively \push or \pull) model, the choice of $c_v$ is done uniformly at random among the neighbors of $v$. This means that uniform models make sense even in anonymous networks, in which nodes are not aware of their identities nor of the identities of their neighbors.\footnote{In the general \gossip\ model in which a node can choose which incident edge to activate, a node must, at least, know its degree and have a way to distinguish between its incident edges.}

In this work, we are mainly interested in   models like  \gossip\ that severely restrict  communication~\cite{AAE07,AD,DGMSS11,Doty14,MosselNT14,PVV09}, both for efficiency consideration and because such
models capture aspects of the way consensus is reached in  biological population systems, and other domains of interest in network
science \cite{AAD+06,AFJ06,Dolev00,cardelli2012cell,Doty14,FHK17,HouseHunt}.
Communication  capabilities in such scenarios are  typically 
constrained and non-deterministic:  both features     are well-captured by uniform models.

Asynchronous variants of the \gossip model (such as \emph{Population Protocols}  \cite{AAE07,AAD+06})  have also been extensively studied  \cite{BGPS06,GNW16,PVV09}. In this variant,  no global clock is available to nodes. Instead, 
nodes are idle until 
a single node   is activated by a (possibly random) scheduler, either in discrete time or in continuous time. 
When a node wakes up, it activates one of its incident edges and wakes up the corresponding neighbor. Communication happens only between those two vertices, which subsequently go idle again until the next time they wake up.

Previous  studies show that, in both \push\ and \pull\ variants of uniform \gossip, 
  (binary) Consensus, Majority Consensus and Broadcast can be solved within logarithmic time (and work per node) in the complete graph, via elementary protocols\footnote{In the case of Majority Consensus, the initial additive bias must have size $\Omega(\sqrt{n \log n})$.}, with high probability (for short \emph{w.h.p.}\footnote{In this paper,  we say that an event $\mathcal{E}_n$ holds
\emph{w.h.p.} if $\Prob{\mathcal{E}_n}\ge 1 - n^{-\alpha}$, for some $\alpha >1$.})
  \cite{AAE07,BCNPST14,BGPS06,DGMSS11,GNW16,KSSV00} (see also
  Section \ref{ssec:othrelatedw}).
  Moreover, efficient protocols have been proposed for   Broadcast and Majority Consensus for some restricted families of graphs such as regular expanders and random graphs \cite{AD12,CLP09,chierichetti_almost_2010,CER14,GSa,MNRS14}.
 
  However, while for Broadcast  $\Omega(\log n)$ time and work are necessary in the complete graph \cite{BGPS06,GNW16,
KSSV00}, prior to this work, it was  still unknown whether a more efficient protocol existed for Consensus and Majority Consensus.

\subsection{Our contribution I:  Broadcast is ``no harder'' than Consensus over noiseless communication} \label{ssec:introlucio}

Our first result is a reduction of a weak form of Broadcast to Consensus (Theorem \ref{thm:mainlowerI}) which establishes, among other lower bounds, tight logarithmic lower bounds for Consensus and Majority Consensus both in the uniform \gossip (and hence uniform \pull and \push as well) model and in the general \push\ model.

To describe our result, it is useful to introduce the notion of nodes \emph{infected} by a source node in a distributed protocol: if $s$ is a designated source node in the network, then we say that at time $0$ the node $s$ is the only infected node and, at time $t$, a node is infected if and only if either it was infected at time $t-1$ or it received a  communication from an infected node at time $t$.

This notion is helpful in thinking about upper and lower bounds for Broadcast: any successful broadcast protocol from $s$ needs to infect all nodes from source $s$, and any protocol that is able to infect all nodes from source $s$ can be used to broadcast from $s$ by appending $\msg$ to each message originating from an infected node. Thus any lower bound for infection is also a lower bound for Broadcast, and any protocol for infection can be converted, perhaps with a small overhead in communication, to a protocol for Broadcast.
For example, in the \push model (either uniform or \general\footnote{See Section \ref{sssec:distsyst} for a formal definition of the two variants.}), the number of infected nodes can at most double at each step, because each infected node can send a message to only one other node, and this is the standard argument that proves  an $\Omega(\log n)$ lower bound for Broadcast.

In Theorem \ref{thm:mainlowerI} we show that lower bounds for infection {\em also give lower bounds for Consensus}. More precisely we prove that if we have a Consensus protocol that, for every initial  opinion vector, succeeds in achieving consensus with probability $1-o(1/n)$, then there is an initial opinion vector and a source such that the protocol infects all nodes from that source with probability at least $(1-o(1))/n$. Equivalently, if we are in a model in which there is no source for which we can have probability, say, $\geq 1/(2n)$ of infecting all nodes with certain resources (such as time, memory, communication per node, etc.), then, in the same model, and with the same resources, every Consensus protocol has probability $\Omega(1/n)$ of failing. For example, by the above argument, we have an $\Omega(\log n)$ lower bound for Consensus in the \push model (because, in fewer than $\log_2 n$ rounds, the probability of infecting all nodes is zero).

The proof uses a hybrid argument to show that there are two initial opinion vectors $\bx$ and $\by$, which are identical except for the initial opinion of a node $s$, such that there is at least a $(1-o(1))/n$ difference between the probability of converging to the all-zero configuration   starting from $\bx$ or from $\by$. Then, we argue that this difference must come entirely from runs of the protocol that fail to achieve consensus (which happens only with $o(1/n)$ probability) or from runs of the protocol in which $s$ infects all other nodes. Thus the probability that $s$ infects all nodes from the initial vector $\bx$ has to be $\geq (1-o(1))/n$.

As for   Majority Consensus,  we have a similar reduction, but from a variant of the infection problem in which there is an initial set of $b$ infected nodes.\footnote{Recall that $b$ is the value such that we are promised that
the majority opinion is held, initially, by at least $n/2 + b$ nodes.}

Lower bounds for infection are known in several models in which there where
no previous negative results for Consensus. We have not attempted to survey 
all possible applications of our reductions, but here we enumerate some of them:

\begin{itemize}
    \item In the uniform \gossip model (also known as uniform \push-\pull model), and  in  the \general\ \push\ model,   tight analysis (see \cite{H15,KSSV00} and Subsection \ref{ssec:appllowless})
    show that any protocol $\sP$ for the complete graph w.h.p.\ does not complete Broadcast  within less than $\beta \log n$ rounds, where $\beta$ is a sufficiently small constant. Combining this lower bound with our reduction result above, we get an  $\Omega(\log n)$ lower bound for Consensus. This is the first known lower bound for Consensus  showing  a full   equivalence between the complexity of Broadcast and Consensus in such models. Regarding Majority Consensus, we also obtain an $\Omega(\log n)$ lower bound for any initial bias $b=O(n^{\alpha})$, with $\alpha <1$.  
    
    \item In a similar way, we are able to prove a lower bound of
     $\Omega(n\log n)$ number of steps (and hence $\Omega(\log n)$ parallel time) or 
     $\Omega(\log n)$ number of messages per node for Consensus  on  an asynchronous variant of the  \gossip model, named \emph{Population Protocols} with uniform/probabilistic scheduler, as defined in \cite{AAE07}.
        
        \item 
         The last application we mention here concerns the synchronous  \emph{Radio Network} model \cite{ABLP91,BGI92,CDGLNN08,SW89}.
        Several optimal bounds have been obtained on the Broadcast time \cite{BGI92,CMS01,KP02,KLNOR10,KM98} while only few results are known for Consensus time   \cite{CDGLNN08,SW89}. In particular, we are not aware of better lower bounds other than the trivial  $\Omega(D)$ (where $D$ denotes the diameter of the network).
        Then, by combining  a previous lower bound in   \cite{ABLP91} on Broadcast with our reduction result, we get a new lower bound for Consensus in this model (see Subsection \ref{ssec:appllowless}).
         
\end{itemize}

We also mention that our reduction allows us to prove that some of  the above lower bounds hold also for a weaker notion of Consensus, namely $\delta$-Almost Consensus (where $\delta n$ nodes are allowed to not agree with the rest of the nodes), and even if the nodes have unbounded memory and can send/receive messages of unbounded size. We will expand on these comments in the technical sections.

\subsection{Our contribution II: Consensus over   noisy communication} \label{ssec:noisy}

We then turn to the study of distributed systems in which the communication links between nodes are noisy. We will consider a basic model of high-noise communication:
the binary symmetric channel \cite{MacK03} in which each exchanged bit is flipped independently at random with probability $1/2 - \epsilon$, where $0 \leq \epsilon < 1/2$, and we refer to $\epsilon$ as the \emph{noise} parameter of the model.

In models such as \local\ and \congest, the ability to send messages of logarithmic length (or longer) implies that, with a small overhead, one can encode the messages using error-correcting codes and simulate  protocols that assume errorless communication.

In the uniform \gossip\ model with one-bit messages, however, error-correcting codes cannot be used and, indeed, whenever the number of rounds is sublinear in $n$, most of the pairs of nodes that ever communicate only exchange a single bit.

The study of fundamental distributed tasks,
such as Broadcast and Majority Consensus, 
has been undertaken in the uniform \gossip\ model with one-bit messages and noisy links 
\cite{BFKN18,FHK17} as a way of modeling the restricted and faulty communication that takes place in biological systems, and as a way to understand how information can travel in such systems, and how they can repair inconsistencies. Such investigation falls
under the agenda of \emph{natural algorithms}, that is, the investigation of biological phenomena from an algorithmic perspective \cite{C12,navlakha2015distributed}.

In  \cite{FHK17}, the authors prove
that (binary) Broadcast and (binary) Majority Consensus can be solved in time $O(\epsilon^{-2} \log n)$, where $\epsilon$ is the noise parameter, in the uniform \push\ model with one-bit messages. They also prove a matching lower bound assuming that the protocol satisfies a certain symmetry condition, which is true for the protocol of their upper bound. This has
been later generalized to non-binary opinions
in \cite{fraigniaud_noisy_2018}. 

In the noisy uniform \pull\ model however, 
\cite{BFKN18} proves an  $\Omega(\epsilon^{-2}n)$ time lower bound\footnote{They actually proved a more general result including non-binary noisy channels.}. 
 This lower bound is proved even under   assumptions that strengthen the negative
  result, such as unique node IDs, full synchronization, and shared randomness (see
 Section 2.4 of \cite{BFKN18} for more details on this point). 
 
 Such a gap between noisy uniform \push\ and \pull\ comes from the fact that, in the \push\ model, a node is allowed to decline to send a message, and so one can arrange a protocol in which nodes do not start communicating until they have some confidence of the value of the broadcast value. In the \pull\ model, instead, a called node must send a message, and so the communication becomes polluted with noise from the messages of the non-informed nodes.
 
 What about Consensus and Majority Consensus in the noisy \pull\ model? Our reduction in Theorem \ref{thm:mainlowerI} suggests that there could be   $\Omega(\epsilon^{-2}n)$  lower bounds for Consensus and Majority Consensus, but recall that the reduction is to the infection problem, and infection is equivalent to Broadcast only when we have errorless channels.

   \subsubsection{Upper bounds in noisy uniform \texorpdfstring{\pull}{PULL}}

   We devise a simple and natural protocol for Consensus for the noisy uniform \pull model having convergence time   $O(\epsilon^{-2} \log n )$, w.h.p., thus exhibiting an exponential gap between Consensus and Broadcast in the noisy uniform \pull\ model.
   
   The protocol runs in two phases. In the first phase, each node repeatedly collects a batch of $O(1/\epsilon^2)$ pulled opinions and then updates its opinion to the majority opinion in the batch. This is done $O(\log n)$ times so that the first phase takes $O(\epsilon^{-2} \log n)$ steps. In the second phase, each node collects a batch of $O(\epsilon^{-2} \log n)$ pulled opinions and then updates its opinion to the majority opinion within the batch.
   
   The main result of the analysis is that, w.h.p., at the end of the first phase there is an opinion that is held by at least $n/2 + \Omega(n)$ nodes, and that if the initial opinions where unanimous then the initial opinion is the majority opinion after the first phase. Then, in the second phase, despite the communication errors, every node has a high probability of seeing the true phase-one majority as the empirical majority in the batch and so all nodes converge to the same valid opinion. 
   
   To analyze the first phase, we break it out into two sub-phases (this breakdown is only in the analysis, not in the protocol): in a first sub-phase of length $O(\epsilon^{-2} \log n)$, the protocol ``breaks symmetry'' w.h.p. and, no matter the initial vector, reaches a configuration in which one opinion is held by $n/2 + \Omega(\sqrt{n \log n})$ nodes. In the second sub-phase, also of length $O(\epsilon^{-2} \log n)$, a configuration of bias  $\Omega(\sqrt{n \log n})$ w.h.p. becomes a configuration of bias $\Omega(n)$.
   The analysis of this sub-phase for achieving Majority Consensus is similar to that in \cite{FHK17,fraigniaud_noisy_2018}.
   If the initial opinion vector is unanimous, then it is not necessary to break up the first phase into sub-phases, and one can directly see that a unanimous configuration maintains a bias $\Omega(n)$, w.h.p., for the duration of the first phase.

  A consequence of our analysis  is that, if the initial opinion vector has a bias $\Omega(\sqrt{n \log n})$, then the  protocol converges to the majority, w.h.p. So, we  get  a  Majority-Consensus protocol for this model under the above condition on the bias.

We also provide a Broadcast protocol that runs in $O(\epsilon^{-2} n\log n)$ steps in the noisy uniform PULL model, nearly matching the $\Omega(\epsilon^{-2} n)$ lower bound mentioned before. This protocol also runs in two phases. In the first phase, which lasts for order of $\epsilon^{-2} n \log n$ steps, the informed node responds to each PULL request with the message, and other nodes respond to each PULL request with zero. After this phase, each node makes a guess of the value of the message, and with high probability the number of nodes that make a correct guess is at least $n/2 + \Omega(\sqrt{n \log n})$. The second phase is a Majority Consensus protocol applied to the first-phase guesses, which, as discussed above, takes only $O(\epsilon^{-2} \log n)$ steps.

\subsubsection{Lower bounds in \texorpdfstring{noisy \pull models}{noisy PULL models}} 
We prove that  any  Consensus  protocol that has
error probability at most $\delta$   requires $\Omega\left(\epsilon^{-2} \log \delta^{-1} \right)$ rounds (Theorem \ref{thm:noisylower}).       
This shows that the complexity $O(\epsilon^{-2} \log n)$ of our protocol described above is tight for protocols that succeed w.h.p.
We remark that our result holds for any version (\general and uniform) of the noisy \pull model with noise parameter $\epsilon$, unbounded local memory,  even assuming unique node IDs.

In  \cite{FHK17}, an $\Omega\left(\epsilon^{-2} \log \delta^{-1} \right)$ round lower 
bound is proved for Majority Consensus in the uniform \push model, for a restricted class of protocols.
Their argument, roughly speaking, is that each node needs to receive a bit of information from the rest of the graph (namely, the majority value in the rest of the graph), and this bit needs to be correctly received with probability $1-\delta$, while using a binary symmetric channel with error parameter $\epsilon$. It is then a standard fact from information theory that the channel needs to be used
$\Omega(\epsilon^{-2} \log \delta^{-1})$ times.

It is not clear how to adapt this argument to the Consensus problem. Indeed, it is not true that every node receives a bit of information with high confidence from the rest of the graph (consider the protocol in which one node broadcasts its opinion), and it is not clear if there is a distribution of initial opinions
such that there is a node $v$ whose final opinion has mutual information close to 1 to the global initial opinion vector given the initial opinion of $v$ (the natural generalization of the
argument of \cite{FHK17}).

Instead, we show that there are two initial opinion vectors $\bx$ and $\by$,   a node $v$, and a bit $b$, such that the initial opinion of $v$ is the same in $\bx$ and $\by$, but the probability that $v$ outputs $b$ is $\leq \delta$ when the initial opinion vector is $\bx$ and $\geq \Omega(1)$ when the initial opinion vector is $\by$. Thus, the rest of the graph is sending $v$ a bit of information (whether the initial opinion vector is $\bx$ or $\by$) and the communication succeeds with probability $\geq 1-\delta$ when the bit has one value and with probability $\geq 1/3$ if the bit has the other value. Despite this asymmetry, if the communication takes place over a binary symmetric channel with error parameter $\epsilon$, a calculation using KL divergence shows that the channel has to be used $\Omega(\epsilon^{-2} \log \delta^{-1})$ times.

The $\Omega(\epsilon^{-2} n)$ lower bound of  \cite{BFKN18} for Broadcast in the uniform PULL model applies to protocols that have constant probability of correctly performing the broadcast operation. In Lemma \ref{lem:broadcastlower} we sketch a way of modifying their proof to derive an $\Omega(\epsilon^{-2} n\log n)$ 
for uniform PULL protocols for Broadcast that have high probability of success, matching the $O(\epsilon^{-2} n\log n)$
round complexity of our protocol mentioned above.

\subsection{Two separations that follow from our bounds}

% As final remark, we want to point out that previous results in \cite{FHK17,BFKN18} and our results establish  interesting complexity gaps between Consensus and Broadcast in presence/absence of noise for the \gossip models.  

We remark that our results establish two interesting separations.

The first, concerns the complexity gap between Consensus and Broadcast in the presence or absence of noise.
Informally, we prove   that, in the noiseless world, Broadcast and Consensus essentially have the same complexity in several natural  models (Corollary \ref{cor:noiselesslower}). 
On the other hand,   we show that there is a natural model where the presence of noise has reasonable motivations (\cite{BFKN18,FHK17}), namely the  noisy uniform \pull, for which the complexity of the two problems exhibits an exponential gap, since in this model Broadcast requires $\Omega(\epsilon^{-2}n)$ rounds \cite{BFKN18} while we prove that Consensus can be solved in $O(\epsilon^{-2}\log n)$ time (Theorem \ref{thm:upperbound}). 

%non ripeterei di nuovo
%Indeed, any protocol for Broadcast requires $\Omega(\frac{n}{\epsilon^2})$ %rounds \cite{BFKN18} while we prove that Consensus can be solved in %$O(\frac{1}{\epsilon^2}\log n)$ time.  

The second fact regards a separation between general \pull and \push models as far as   Consensus   is concerned in the noiseless world. Indeed, if we assume unique IDs, in the general \pull model, Consensus can be easily solved in constant time: every node can copy the opinion of a prescribed node by means of a single pull operation.
On the other hand, in the general \push model, our Broadcast-Consensus reduction 
shows that $\Omega(\log n)$ rounds are actually necessary for solving Consensus.

% Even though our study  is of theoretical nature and  work on models that neglect several aspects of natural processes,
% the current   set of results on this topic  might suggest   that   efficient and fast coordination observed in some  natural population systems, prone to communication noise,
% might be determined by consensus mechanisms which do not require  rumor-spreading processes started by  few leaders but, rather, from 
% a more fairly-distributed wake-up. We believe our result can provide useful suggestions to drive specific experiments  and tests on real 
% natural population systems, like the ones performed in \cite{BFKN18}.

% \begin{figure}
%     \centering
%     \includegraphics[scale=1
%     ]{trunk/table.pdf}
%     \caption{A comparison between Broadcast and Consensus in presence/absence of noise in the uniform \pull model.}
%     \label{fig:table}
% \end{figure}

\subsection{Other related work} \label{ssec:othrelatedw}
Consensus and Broadcast are fundamental algorithmic problems which have been the subject of a huge number of studies focusing on several distributed models and different computational aspects \cite{Dijkstra74,Dolev00,PSL80,R83}.
We here    briefly discuss those results which are more relevant w.r.t. to our contribution. 
     
%\subsection{\texorpdfstring{\gossip Models}{GOSSIP Models}}
%\label{ssec:gossiprelated}

%We provide a brief overview of the main related work in the %(\general) \gossip\  model \cite{DGHILSSST87,H15}, by %distinguishing the noiseless and noisy case.

\paragraph{Noiseless communication.}
 Classical results prove that on the uniform \push or \pull models, \emph{Rumor Spreading} (Broadcast) takes logarithmic time \cite{frieze_shortest-path_1985,KSSV00,pittel_spreading_1987}.  Then, a series of recent works 
has shown that simple uniform  \pull protocols   can quickly achieve   Consensus, Majority Consensus and Broadcast even  in the presence of a bounded number of   node crashes or Byzantine nodes \cite{PetraAL17,BCNPS15, BCNPST14, becchetti2016stabilizing,DDMM16,DGMSS11,GL17,KSSV00}.  
The logarithmic bound is known to be tight for  Broadcast  \cite{KSSV00},  while, as remarked earlier,  no non-trivial lower bounds are known for Consensus   in any variant of the \gossip model. 
Further  bounds are known for Broadcast and Majority Consensus   on graphs having good expansion properties (see for instance \cite{CLP09,chierichetti_almost_2010,chierichetti_rumour_2010,giakkoupis_tight_2011, GSa, giakkoupis_tight_2014}).
As for the \general\ \gossip model with special conditions on  node IDs, we mention the 
upper bound $O(\sqrt{\log n})$ obtained in \cite{AE17} which has been then improved to $\Theta(\log\log n)$ bound in \cite{HM14}.
  A further  issue is the minimum amount of \emph{randomness} necessary to solve Broadcast within a given time bound. In the \push model, this issue is investigated  in \cite{doerr_quasirandom_2008,doerr_quasi_random_2011}, where  upper bounds and tradeoffs are given.

\paragraph{Noisy communication.}
In Subsection \ref{ssec:noisy} we introduced and motivated the noisy communication model studied in \cite{BFKN18,FHK17,fraigniaud_noisy_2018} and adopted in this paper.
Another  model of noisy communication   for 
distributed systems is the one considered in \cite{ADHS17,CGH17}. Departing significantly from the model we adopt in this paper, here there is a (worst-case) adversary that can adaptively flip the  bits exchanged during the execution of any   protocol and the goal is to provide a robust version of the protocol   under the assumption that the adversary has a limited budget on the number of   bits it can change. Efficient solutions for such models  typically use silent rounds \cite{ADHS17} and error-correcting codes \cite{ADHS17,CGH17}. 
In \cite{E84} a different  task is studied in a model with  noisy interactions:   all $n$ nodes of a network hold a bit and they wish to transmit to a single receiver. This line of research culminated in the $\Omega(n \log \log n)$ lower bound on the number of messages shown in \cite{GKS08}, matching the upper bound shown in  \cite{G88}.

 \paragraph{Other communication models.}

In \cite{FP15},  Consensus   has been studied on a fault-free model. They provide bounds on the message complexity for deterministic protocols.
  As for Radio Networks, we have already discussed  the results for  static topologies \cite{ABLP91}.  We remark here that finding lower bounds for Consensus (and Leader Election) on a rather general model of \emph{dynamic} Radio Networks is an open question posed in \cite{KNR10}, where some lower bounds on the \emph{$k$-Token Dissemination} Problem (a variant of   Broadcast) have been derived. Another dynamic model of Radio Networks where lower bounds on Broadcast time have been derived can be found in \cite{CMPS07}. Even though we have still  not verified the applicability of our reduction result in these contexts, we believe this  might be  possible.
  Finally, we mention the works   \cite{CDGLNN08,SW89} that  consider faulty models (some with interference detectors), and provide complexity bounds on Consensus.

\iffalse

Emanuele
"The line of research initiated by El-Gamal [9], also studies a broadcast problem with noisy interactions. The regime however is rather different from ours: all n agents hold a bit they wish to transmit to a single 4receiver. This line of research culminated in the Ω(n log log n) lower bound on the number of messages shown in [13], matching the upper bound shown many years earlier in [12]."

\fi
  
\iffalse 
ClassicClassic research on Consensus focused on faulty models where   nodes and/or links are prone to crashes or Byzantine behaviours. This was also the   framework   the problem has been introduced on
\cite{Dijkstra74,Dolev00,PSL80,R83}.
\fi

\subsection{Roadmap of the Paper}
The rest of the paper is organized as follows.
In Section \ref{sec:prely}, 
  preliminary definitions are given which will be used all over the paper.
  Section \ref{sec:noiseless} deals with the noiseless case. In particular, we first describe the general reduction result of Broadcast to Consensus in   noiseless communication models and derives
  the main applications of this result to some specific models, and then, in Section \ref{sec:majoritylower}, we give the simple reduction of 
  (multi)-Broadcast to Majority Consensus showing that the latter requires
  logarithmic time in any noiseless uniform \gossip model.
  Section \ref{sec:noisylb} provides the lower bound on the noisy \pull model obtained by a reduction to an asymmetric Two-Party Protocol.
  In Section \ref{sec:noisyub}, we propose a simple majority protocol
  and show it solves Consensus and Majority Consensus in the 
  noisy uniform \pull model within $O(\epsilon^{-2}\log n)$ rounds. We also describe a protocol for Broadcast in the noisy uniform \pull running in $O(\epsilon^{-2} n \log n)$ rounds.
  Finally, some technical tools 
  are located in a separate appendix.

 \section{Preliminaries}
 \label{sec:prely}

\subsection{Distributed systems and communication models}
\label{sssec:distsyst}
Let $\sS$ be a distributed system formed by a set $V$ of $n$   nodes which mutually interact by exchanging messages over  a  connected  graph $G=(V,E)$ according to a fixed communication model $\sM$.
The definition of  $\sM$ includes all features of node communications including, for instance,  synchronicity or not and the presence of link faults. 

A configuration $\confg$ of  a distributed system $\sS$ is
the description of the states of all the nodes  of $\sS$ at a given time. 
If we execute a protocol $\sP$ for $\sS$, 
the random configuration the system lies in a generic  time is denoted as
$\Confg$. 

When, in the \gossip, \push or \pull models, at each round the communication is established with a random neighbor chosen independently and u.a.r., we call the communication model \textit{uniform}. 
In order to remark the difference with the uniform case, we  call  the communication model \emph{\general} when nodes are equipped with unique IDs which are known to all neighbors, and each node can choose the identity of the neighbor with which to communicate (possibly in a random way).

Finally, we distinguish two main communication scenarios.  
In the \emph{noiseless} models, every transmitted
message on a link of the graph is received safely, without any error. 

In the presence of \emph{communication noise}, instead,
each  bit of any transmitted  message is flipped independently at random with probability $1/2-\epsilon$, where $\epsilon \in (0,1/2]$ is the \emph{noisy} parameter. 
Then, in the sequel,  the version of each   model $\sM$, in which the presence of    communication noise  above is introduced, will be shortly denote as \emph{noisy} $\sM$.
Notice that, in order to capture the role of noise in systems where standard error correcting
codes techniques are not feasible, we consider models where each single point-to-point transmission consists of one bit only. In this way, we easily have that, in \gossip models, the bit-communication and the convergence time of a Protocol are strongly related.

\subsection{Consensus and Broadcast   in distributed systems}
\label{ssec:cons&broad}

Several versions of Consensus have been considered in the literature \cite{Dijkstra74,Dolev00,PSL80,R83}.
Since our interest  is  mainly focused  on  models having  strong constraints on communication (random, limited, and noisy),
we adopt 
some weaker, probabilistic variants of consensus, studied
in \cite{AAE07,becchetti2016stabilizing,DGMSS11}, that well captures  this  focus. 
 
Formally, we say 
a protocol guarantees  (binary) \emph{Consensus} if, starting
from any initial node opinion vector $\bx \in \{0,1\}^n$, the system reaches  w.h.p.    a   configuration where every node has the same opinion (\emph{Agreement}) and this opinion is \emph{valid}, i.e., it was supported by at least one node at the starting time. Moreover,
once the system  reaches this \emph{consensus} configuration, it  is required to stay there
for   any arbitrarily-large polynomial number of rounds, w.h.p. 
This \emph{Stability} property somewhat replaces the  \emph{Termination} property required by other,   classic notions of consensus introduced in stronger distributed models \cite{lynch_distributed_1996}.
%:
% we are here interested in scenarios where 
% nodes represent simple and anonymous computing units which are not necessarily able to detect %any global property.

In order to define Majority Consensus, we need to introduce the notion of bias of an (initial) opinion vector.
Given any vector $\bx \in \{0,1\}^n$, the \emph{bias} $\bias$ associated to $\bx$ is the difference between the number of nodes supporting opinion $1$ and the number of those supporting $0$ in $\bx$. The state of each node clearly depends on the specific protocol, however, we can always assume   it contains the current opinion of the node, so, we can also define the bias of a configuration $s(\confg)$. When the opinion vector (or the configuration) is clear from the context, we will just use the term $s$. The \emph{majority opinion} in a given vector (configuration) $\bx$ is the one having the largest number of nodes supporting it. With the term \emph{majority}, we will indicate the number of nodes supporting the majority opinion.
A   protocol guarantees (binary) \emph{Majority Consensus} if, starting
from any initial  opinion vector $\bx \in \{0,1\}^n$ with bias $s(\bx) >0$,
 the system reaches  w.h.p.\    a   configuration where every node has the same
opinion   and this opinion is the initial majority one. Moreover, we require the same stability property we define  for Consensus.

Both the notions of Consensus and Majority Consensus above can be further relaxed to those of $\delta$-\emph{Almost Consensus} and $\delta$-\emph{Almost Majority Consensus}, respectively. 
According to such weaker notions, we allow the system to converge to 
an almost-consensus regime where    $\delta n$      \emph{outliers} may have a different opinion from the rest of the nodes.
In this case, the fraction of outliers is a performance parameter of the protocol we will specify in the statements of our results.
Even in this weaker notion, we require the same  property of  Stability but we remark that the subset of outliers may change during the almost-consensus regime. 
We emphasize that all lower bounds we obtain in this paper holds
for such weaker versions of Almost Consensus, while the upper bound in Section
\ref{sec:noisyub} refers to (full) Consensus.

As discussed in the introduction, our work also deals with    
the (single-source) Broadcast task (a.k.a. \emph{Rumor Spreading}).
Given any source node $s\in V$ having an initial  message $\msg$, a \emph{Broadcast Protocol}
$\sP$  is a protocol that,  after a finite number of rounds, makes  every  node in $V$   receive a copy of  (and, thus,  be \emph{informed} about)    $\msg$, w.h.p.\footnote{The success probability of the protocol is here defined over both the random choices (if any) of the communication mechanism and the ones of $\sP$.}.
Similarly to Consensus, we also consider a weaker  version of Broadcast
%named $\delta$-\emph{Almost Broadcast}
where the  final number of informed nodes is required to be at most $(1-\delta)n$, w.h.p.

 \iffalse
\noindent
\textbf{Remark.}
A further crucial property of a  consensus protocol is 
\textit{self-stabilization} \cite{becchetti2016stabilizing,DGMSS11,PSL80}: Informally, if the system is ``perturbed'' by some external event and moved to an arbitrary configuration, then the protocol must bring the system back to a valid consensus  and, moreover,  once the system reaches consensus, it must remain in that configuration forever, unless a further  external event takes
place\footnote{Notice that, according to previous work \cite{becchetti2016stabilizing,DGMSS11}, 
we require self-stabilization to hold \emph{with high probability}.}.  Self-stabilizing  consensus processes  are   fundamental building-blocks that play an important role in coordination tasks and self-organizing behavior in population
systems~\cite{cardelli2012cell,C12,Doty14,NB11}. 
Another important version of  the consensus problem is   \emph{Majority Consensus} \cite{AAE07,BCNPST14,DGMSS11}.
 In this version, assuming the initial configuration has
some non-negligible \emph{bias} towards one of the possible opinions\footnote{According to previous literature,
 the bias of a binary configuration is the absolute difference
  between the number of nodes supporting $0$ and the number of those supporting 1.},
 a majority-consensus protocol must let the system converge and stabilize to 
  the monochromatic configuration where all nodes support the majority opinion. 
 \fi

\section{Noiseless Communication: Broadcast vs Consensus} \label{sec:noiseless}
In this section we provide our first main result (Theorem~\ref{thm:mainlowerI})
which establishes a strong connection between (Almost) Consensus and a
suitable, weaker version of (Almost) Broadcast in the  noiseless-communication
framework. We first describe the result in a rather general setting and then we
show its consequences, namely some lower bounds for the (Almost) Consensus
problem in specific communication models. 

Notice that in Section~\ref{sec:majoritylower}, we complement
Theorem~\ref{thm:mainlowerI} with an analogous lower bound for Almost-Majority
Consensus with sub-linear initial bias (see Theorem~\ref{lem:majlower}).  

Let $\sS$ be a distributed system formed by a set $V$ of $n$ nodes which
mutually interact over a support graph $G = (V,E)$ according to a fixed
communication model $\sM$. The crucial assumption we make in this section on
$\sM$ is the absence of communication noise (i.e. message corruption): whenever
a node $v$ transmits a message on one of its links, either this message is
received with no change or it is fully lost and, in the latter case, both
sender and receiver cannot get any information from the state of the
corresponding port (no fault detection).

Under the above noiseless framework, the next theorem essentially states
that (Almost) Consensus cannot be ``easier'' than (Almost) Broadcast. As we
similarly show in Section~\ref{sec:noisylb}, much of the technical difficulty
in reasoning on the valid-consensus problem arises from the high level of
freedom nodes have in agreeing on the final consensus value, since both values
are valid solution as long as not all nodes start with inputs that are already
identical. 

In order to state the reduction, we need to introduce a slightly-different
variant of Broadcast where, essentially, it is (only) required that \emph{some}
information from the source is spread on the network. Formally 

\begin{definition}\label{def:infection}
A protocol $\sP$ solves the $\gamma$-\emph{Infection} problem w.r.t. a source
node $s$ if it \emph{infects} at least $\gamma n$ nodes, where we define a node
\emph{infected} recursively as follows: initially only $s$ is infected; a node
$v$ becomes infected whenever it receives \emph{any} message from an infected
node.
\end{definition}

Notice that a protocol $\sP$ solving the $\gamma$-Infection problem w.r.t. a
source node $s$ can be easily turned into a protocol for broadcasting a message
\msg\ from $s$ to at least $\gamma n$ nodes. Indeed, we give the message \msg\
to the source node $s$, and we simulate $\sP$. Every time an infected node
sends a message, it appends \msg\ to it. Clearly, the size of each message in
$\sP'$ is increased by the size of \msg.

The next theorem is the main result of the section. Informally, it states that
any protocol for Consensus actually solves the Infection problem (when
initialized with a certain opinion vector) in a weak sense: the infection is
w.r.t. a source that depends on the consensus protocol in a (possibly)
uncontrolled manner; and (ii) the success probability of the infection is quite
low. Another intuitive way to look at the result is as follows: any consensus
protocol needs to solve the Infection problem from a certain source node when
it starts from a certain initial opinion vector.

\begin{theorem}
    \label{thm:mainlowerI}
    Let $\sP$ be a protocol reaching $\delta$-Almost Consensus with probability at
    least $1 - o(1/n)$. Then, a source node $s$ and an initial opinion vector $\mathbf{x}$ exist such that $\sP$, starting from $\mathbf{x}$, solves the
    $(1-2\delta)$-Infection problem w.r.t. $s$ with probability at least $(1 -
    o(1))/n$.
\end{theorem}

\begin{proof}
Let $V = \{1,2, \dots,n\}$ be an arbitrary ordering of the vertices and, for
any $k = 0, \ldots, n$, let $\bx_k$ be the initial opinion vector in which the
first $k$ nodes start with $0$ and the other $n-k$ with $1$. Moreover, let
$Z_k$ be the indicator random variable taking value $1$ when $\sP$, starting
from $\bx_k$, reaches $\delta$-Almost Consensus on value $1$, otherwise $Z_k$
takes value $0$ (note that $Z_k = 0$ also when the protocol fails to reach $\delta$-Almost Consensus). 
    
Since the protocol converges to an almost consensus with probability at least
$1 - o(1/n)$, it must hold that 
\[ 
    \Ex{Z_0} \geqslant 1 - o\left(\frac{1}{n}\right) 
    \ \mbox{ and } \ 
    \Ex{Z_n} \leqslant o\left(\frac{1}{n}\right),
\]
and
\[
    \sum_{i = 1}^n \left(\Ex{Z_{i-1}} - \Ex{Z_i}\right) = \Ex{Z_0} - \Ex{Z_n}
    \geqslant 1 - o\left(\frac{1}{n}\right)\,.
\]
Hence, a node $k^*$ exists such that
\begin{equation}
    \label{eq:kappastar}
    \Ex{Z_{k^* - 1}} - \Ex{Z_{k^*}} 
    \geqslant \frac{1 - o(1)}{n} \,. 
\end{equation}
We now show that, when $\sP$ starts from opinion vector $\bx_{k^*}$, $\sP$ is
(also) solving the $(1-2\delta)$-Infection problem w.r.t. source node $k^*$.

First observe that $\Ex{Z_{k^* - 1}} - \Ex{Z_{k^*}} \leqslant \Prob{Z_{k^*-1}
= 1 \; \wedge \; Z_{k^*} = 0}$ and let us name $I_{k^*}$ the set of nodes
infected by node $k^*$, starting from opinion vector $\bx_{k^*}$. 
We will prove that, if $Z_{k^*-1} = 1$ and $Z_{k^*} = 0$ then either $|I_{k^*}| 
\geq (1-2\delta) n$ or $\sP$ fails to reach $\delta$-Almost Consensus starting
from $\bx_{k^*}$. More formally, if we name $\mathcal{F}$ the event
\[
\mathcal{F} = \left\{ \sP \mbox{ fails to reach $\delta$-Almost Consensus from
$\bx_{k^*}$}  \right\}
\]
we claim that  
\begin{equation}
\label{eq:implication}
    \{ Z_{k^*-1} = 1 \; \wedge \; Z_{k^*} = 0\}
    \, \Longrightarrow \;
    \{``|I_{k^*}| \geq (1-2\delta) n" \; \vee \; \mathcal{F} \}\,. 
\end{equation}
In order to obtain~(\ref{eq:implication}) we equivalently show that, if 
$Z_{k^*-1} = 1$, $Z_{k^*} = 0$, and $\sP$ does not fail starting from
$\bx_{k^*}$, it must hold that $|I_{k^*}| \geq (1-2\delta) n$. First observe
that, if $Z_{k^*-1} = 1$ then $\sP$ reaches almost consensus on $1$ starting
from $\bx_{k^*-1}$, therefore the number of nodes which output $1$ is at
least $(1-\delta)n$.  Moreover, if $Z_{k^*} = 0$, and $\sP$ does not fail
starting from $\bx_{k^*}$, then the number of nodes which output $1$ is at
most $\delta n$. Thus, moving the system from input $\bx_{k^*-1}$ to
input $\bx_{k^*}$, the number of nodes switching their opinion from $1$ to
$0$ must be at least $(1-2\delta)n$. Finally observe that, since $\bx_{k^*-1}$
and $\bx_{k^*}$ differ only at node $k^*$, the nodes that change their output
value must be infected by $k^*$ (according to Definition~\ref{def:infection}).
Hence, $|I_{k^*}| \geq (1-2\delta) n$, which implies that, when starting from
$\bx_{k^*}$, $\sP$ is also solving the $(1-2\delta)$-Infection problem w.r.t.
node $k^*$. 

To conclude the proof, it remains to bound (from below) the probability with
which this infection happens. From~(\ref{eq:kappastar}), (\ref{eq:implication})
and the union bound, it follows that
\begin{align*}
    \frac{1 - o(1)}{n} 
    & \leq \Ex{Z_{k^* - 1}} - \Ex{Z_{k^*}} 
    \leq \Prob{Z_{k^*-1} = 1 \wedge Z_{k^*} = 0} \\
    & \leq \Prob{``|I_{k^*}| \geq (1-2\delta) n" \vee \mathcal{F}}
    \leq \Prob{|I_{k^*}| \geq (1-2\delta) n} + o\left(\frac{1}{n}\right)\,.
\end{align*}
Thus,
\[
\Prob{|I_{k^*}| \geq (1-2\delta) n} \geq \frac{1 - o(1)}{n}\,.
\]
\end{proof}

\paragraph{Remark.} Observe that factor $2$ in the parameter $\gamma = (1 -
2\delta)$ in the statement of Theorem~\ref{thm:mainlowerI} is tight. Indeed,
consider a protocol in which each node outputs its input value: such protocol
is trivially a $(\delta = \frac{1}{2})$-Almost Consensus protocol, while the
number of infected node has size at most $1$.

\subsection{Specific lower bounds for Consensus} \label{ssec:appllowless}

Theorem~\ref{thm:mainlowerI} allows us to derive lower bounds for Consensus for
specific communication models and resources by using lower bounds for the
Infection problem. In fact, by simply restating Theorem~\ref{thm:mainlowerI},
we obtain the following.

\begin{corollary}
    \label{cor:noiselesslower}
    Let $\M$ be any fixed resource (e.g. time, work, bit-communication, etc.)
    defined on a distributed system $\sS$, and suppose that any protocol, which
    uses at most  $\brotau$ units of $\M$,  fails to solve the $(1-
    2\delta)$-Infection problem w.h.p. from any source node.  Then, any  protocol
    on this model reaching $\delta$-Almost Consensus w.h.p., must use at least
    $\brotau$ units of $\M$.
\end{corollary}

We now apply the above corollary in different settings. Unless differently
stated, all results in this subsection refers to the complete graph of $n$
nodes.

\paragraph{The \gossip model.}
The first lower bound is on the Consensus time for the uniform \gossip model (and hence for the uniform \push and the uniform \pull as well). We first 
state a simple technical result which will be handy also in the proof of Corollary \ref{cor:gossipmajority}. 
% This is a well-known result in the community, however, for the sake of completeness, its proof is given in Appendix~\ref{apx:omitted}. 
This is a well-known result in the community, however, for the sake of completeness, we give a self-contained proof.

\begin{lemma}
    \label{lem:gossipbroadcast}
    Consider  the uniform \gossip model and  fix 
    any constants $\alpha$ and $\gamma$ such that $0<\alpha,\gamma<1$. Then,  there is a sufficiently small constant $\beta_{\alpha,\gamma}>0$   such that, starting from any subset of infected nodes of size $O(n^\alpha)$,   the   $\gamma$-\emph{Infection} problem requires at least   $\beta\log n$   rounds, w.h.p.
\end{lemma}

\begin{proof}
    The proof shows that, starting from any subset of $O(n^\alpha)$ infected nodes,  the set of infected nodes  grows by at most a constant factor at each round,w.h.p. 
    The latter fact easily implies that, if $\beta$ is sufficiently small, then, within $\beta \log n$ rounds, at most $\gamma n$ nodes are infected w.h.p. 
    
    In order to show that the set of infected nodes increases by at most a constant factor, let $V$ be the set of nodes and $I^{(t)}$ the set of
    infected nodes at time $t$. At each round the number of
    (bidirectional) communication-edges between $I^{(t)}$ and the uninfected nodes $V/I^{(t)}$ is $\frac 1n (n-|I^{(t)}|)|I^{(t)}|\leq  |I^{(t)}|$. By a simple application of
    Chernoff bounds, it follows that  the growth of number of infected
    nodes is bounded by a constant multiplicative factor $\eta$, w.h.p. 
    Since $|I^{(0)}|  \leq n^{\alpha}$ and w.h.p. $|I^{(t+1)}| \leq \eta
    |I^{(t)}|$, then $|I^{(t)}| \leq \eta^{t} \cdot n^{\alpha}$ and the
    latter  is smaller than $\delta n$ as long as $t \leq (1-\alpha)  \log_\eta (\delta n)$. 
    
    Notice that to get a concentration  result over all the process, we just observe that if every event in some family $\{A_i\}_i$ holds w.h.p., then, using the union bound,  the intersection of any polylogarithmic number of such events holds w.h.p.
    
\end{proof}

We can combine the above lemma (in the case where just one source node is   initially infected), 
with Corollary \ref{cor:noiselesslower},  and   get the following. 
\begin{corollary}
    \label{cor:gossiplower}
    Consider the uniform \gossip model and fix any constant $\delta$ such that $0\le \delta<0$. Then, any protocol reaching $\delta$-Almost Consensus w.h.p. requires  $\Omega(\log n)$ communication rounds. 
\end{corollary}

Next,  consider the \general\  \push model and fix any (Broadcast) protocol: then, starting from any source node,  the set of infected nodes  grows by a factor at most $2$ at each  round. 
Hence,  Corollary \ref{cor:noiselesslower} also implies an $\Omega(\log n)$  lower bound for the \general\ \push model. 
\begin{corollary}
    \label{cor:pushlower}
    Let $0\le \delta<0$ be any constant. Then any protocol reaching $\delta$-Almost Consensus w.h.p.\ on the \general\ \push model, requires $\Omega(\log n)$ communication rounds, even when the nodes have unique identifiers. 
\end{corollary}

\begin{remark}
    Corollary \ref{cor:pushlower} should be contrasted with the fact that in the \general\ \pull model, assuming unique identifiers,   (valid) Consensus   can be solved in a single round, by having all nodes adopt the input value of a specific node\footnote{On the other hand, if we assume that nodes do not initially share unique identifiers (for example, in the \pull Model with numbered ports), it is easy to see that the broadcast problem cannot be solved w.h.p. in $o(\log n)$ time, since the number of nodes from which a given node $v$ can receive any information from, increases by at most a factor $2$ at each round.}. 
\end{remark}

\paragraph{Population Protocols.}
Another interesting model where we can apply our reduction is the Population Protocol one with uniform/probabilistic scheduler as defined in \cite{AAE07,AR09}. 
Broadcast in this model has essentially the same complexity of that in  the asynchronous uniform \gossip model: in particolar,  a similar result to that in Lemma \ref{lem:gossipbroadcast} holds for Broadcast time  (see for instance \cite{GNW16}): 
  any  Population Protocol with uniform/probabilistic scheduler on the complete graph cannot infect more than $(1-\delta)n$ nodes  w.h.p. within  $ \beta n \log n$ number of  rounds (and hence $ \beta \log n$ parallel time) or $ \beta \log n$ number of messages per node. Hence, by using Corollary \ref{cor:noiselesslower}, we can state the following result.

\begin{corollary}
    \label{cor:populationlower}
    Let $0\le \delta<0$ be any constant. Then any   Population Protocol (with uniform/probabilistic scheduler) reaching $\delta$-Almost Consensus requires  $\Omega(n\log n)$ number of steps (and hence $\Omega(\log n)$ parallel time) and $\Omega(\log n)$ number of messages per node. 
\end{corollary}

\paragraph{Radio Networks.}
In  the synchronous  \emph{Radio Network} model \cite{ABLP91,BGI92,CDGLNN08,SW89}, the presence of message collisions on the (unique) shared  radio frequency is modelled by the following communication paradigm:   a node can receive a message at a given round $t$ if and only if exactly one of its neighbors transmits at round $t-1$. We consider the model setting
with no \emph{collision detection}, i.e., the  nodes of the graph
are not able to get any information when a collision occurs.

     In \cite{ABLP91} the authors derive a lower bound $\Omega(\log^2n)$ on the radio-broadcast time in networks of constant diameter. In particular, their proof relies on a construction of a family of graphs of $n$ nodes having diameter 2 where every  protocol that runs for no more than $\beta\log^2n$ rounds (where $\beta>0$ is a sufficiently small constant) cannot infect at least one node  w.h.p. (in fact, with probability 1). We observe that the proof can be adapted in order to hold   for any choice of the source and  when every node knows the graph topology, so for any choice of the initial  configuration. This implies that their lower bound also applies on the time required by any  protocol to $\gamma$-infect all nodes (with $\gamma =1$) according to Definition \ref{def:infection}.
     Then,    from   Theorem \ref{thm:mainlowerI}  we get a lower bound on the Consensus time in Radio Networks.
     
     \begin{corollary}
    \label{cor:randralower}
    Consider the Radio Network model. There is a family of constant-diameter graphs, where any  (randomized) protocol reaching Consensus requires  $\Omega(\log^2 n)$  time, w.h.p.
\end{corollary}
   
  \iffalse 
     % LA PROVA DEL LOWER BOUND \generalO SU RADIO BROADCAST RICHIEDE UN CERTO %NODE LABELING  ED UNA     SPECIFICA DELLA SORGENTE SCELTI AVVERSARIALMENTE %RISPETTO AL PROTOCOLLO DI BROADCAST. NON MI E' CHIARO QUINDI COME APPLICARE
     %QUESTO LOWER BOUND. IL PROBLEMA E' ANCHE CHE NON SIAMO SU UN GRAFO COMPLETO.
     
     The bound above holds for any randomized protocol while, for deterministic protocols, a lower bound $\Omega(n\log n)$ on the Broadcast time has been obtained in \cite{CMS01} and its proof ensures the conditions required by our reduction result. In details, We thus get the following
     
\begin{corollary}
    \label{cor:detralower}
    Consider the Radio Network model. There is a family of   graphs, where any  deterministic  protocol reaching Consensus requires  $\Omega(n\log n)$  time.
\end{corollary}

\fi

\subsection{A lower bound for \texorpdfstring{$\delta$-}\emph{Almost Majority Consensus}}
\label{sec:majoritylower}

The   conditions required by Majority Consensus are much stronger than the   validity one and  make the relationship 
between this task and  the  $\gamma$-Infection problem
with multiple source nodes rather simple to derive.

\begin{lemma}\label{lem:majlower}
    Let $\M$ be any fixed resource defined on a distributed system
    $\sS$ and suppose there is no  Infection  protocol  that, starting from any subset of $n^{\alpha}$ nodes with $\alpha<1$, can  inform  at least $(1-\delta)n$ nodes by using at most $\brotau$ units of $\M$, w.h.p. 
    Then, any protocol $\sP$ on this model, reaching $\delta$-\emph{Almost Majority Consensus} w.h.p.,  must use at least $\brotau$ units of $\M$.   
\end{lemma}
\begin{proof}
    W.l.o.g., let $n-b$ be an even number where $b$ is the initial bias.
    Consider an arbitrary labeling of the nodes $v_1, \dots, v_n$ and two
    initial input vectors  $\bx_0$ and $\bx_1$ such that
    \begin{equation}
         \bx_i = 
         \begin{cases}
             v_j=0 & \text{ for }j\in \left\{ 1,\dots,\frac{n-s}2 \right\},\\
             v_j=1 & \text{ for }j\in \left\{ \frac{n-s}2 + 1,\dots, n-s
         \right\},\\
             v_j=i & \text{ for }j\in \left\{ {n-s} + 1,\dots, n \right\}.\\
         \end{cases}
         \label{eq:configMajLower}
    \end{equation}
    In order for a node $v$ to converge to the correct majority opinion, it is
    necessary that it is able to 
    distinguish between configuration $\bx_0$ and $\bx_1$. 
    Since
    $\bx_0$ and $\bx_1$ are identical for all nodes $v_1, \dots, v_{n-s}$, it is
    then necessary for $v$ to be infected by  each of the  \emph{source} nodes   $v_{n-s+1}, \dots, v_{n}$ and the proof is completed.

 \end{proof}

\paragraph{A specific lower bound for Majority Consensus.}

The above lemma allows us to obtain a logarithmic lower bound on the convergence time
required by any  almost Majority-Consensus protocol on    uniform \gossip (and, hence, on uniform \pull and uniform \push).
Notice that the bound holds even for protocols
achieving the task with constant probability only.

\begin{corollary}
\label{cor:gossipmajority}
    Consider    $\delta$-\emph{Almost Majority Consensus} in    uniform \gossip    starting with    
      initial   bias $b\leq n^{\alpha}$, for any   positive constant
    $\alpha <1$. Then, any protocol that 
    solves the task above  with probability at least $2/3$ requires  $\Omega_{(1-\alpha)}(\log n)$ rounds. 
\end{corollary}
\begin{proof}
    The proof  follows from  Lemma \ref{lem:gossipbroadcast} and Lemma \ref{lem:majlower}, where $\M$ is the number of rounds in the  uniform \gossip model.
\end{proof}

%\subsection{Notes on a Lower Bound for Consensus}
\section{Lower Bounds in the Noisy Model} \label{sec:noisylb}

The main result of this section is the following  
lower bound for the $\delta$-Almost Consensus Problem in the noisy \general\ \pull model (and hence in noisy uniform \pull as well). 
%and then, in Subsection \ref{ssec:lbgeneral}, we discuss some %possible generalization of this result.  

\begin{theorem}
    \label{thm:noisylower}
    Let $\delta$ be any real such that $0< \delta < 1/8$ and  consider any protocol $\sP$ for the noisy \general\ \pull model with noise parameter $\epsilon$. If $\sP$ solves  $\delta$-Almost Consensus with probability at least $1 - \delta$, then it requires at least $t = \Omega(\epsilon^{-2} \log \delta^{-1})$ rounds\footnote{We notice the double role parameter $\delta$ has in this statement.}.
\end{theorem}
\begin{proof}[Proof Outline.]
%of Theorem \ref{thm:noisylower}.]
    %
    W.l.o.g.\ we assume that during the execution of $\sP$, every node pulls another node at each round. Indeed, if not true, we can simply consider a protocol $\sP'$ where this property holds but the extra messages are ignored, obtaining an equivalent protocol.
    
    Suppose that we have a  protocol that solves the $\delta$-Almost Consensus with probability at least $\geq 1-\delta$. The definition of $\delta$-Almost Consensus implies that each node has an initial one-bit input and produces a one-bit output and
    
    \begin{itemize}
        \item If the initial opinion vector is all-zeroes, then with probability $\geq 1-\delta$ all nodes but at most $\delta n$ output zero.
        \item If the initial opinion vector is all-ones, then with probability $\geq 1-\delta$ all  nodes but at most $\delta n$ output one.
        \item For every initial opinion vector, with probability $\geq 1-\delta$ all nodes but at most $\delta n$ agree.
    \end{itemize}
    
    From the above constraints, we derive the existence of (at least)  one  node $v^*$ which must get  from the rest of the system ``enough information'' in order to decide its output.  More in details, in Lemma \ref{lem:reduction}, we show that any Almost-Consensus protocol implies a solution to a two-party communication problem over a noisy communication channel, where one of the two parties represents node $v^*$ and it needs to act differently according to the information owned by the other party, namely the rest of
    the graph.
    
    In  Lemma \ref{lem:kl}, we then show  that, in the two-party   problem above, $v^*$ has to receive at least $\Omega(\epsilon^{-2}\log \delta^{-1})$ bits (and thus, according to the noisy \pull, performs
    at least the same number of rounds) in order to recover  the information owned by the rest of the graph with a sufficiently large probability, and thus deciding its output. This implies the desired bound. 
\end{proof}
%
%If the protocol works in a communication model such that
%\begin{itemize}
%\item each node has probability $\geq 1-\delta$ of sending and receiving at most $t$ bits of messages,
%\item each bit of communication from one node to the other passes through a binary symmetric channel that flips the bit with probability $\frac 12 - \epsilon$,
%\end{itemize}
%we will prove that $t \geq \Omega(\epsilon^{-2} \log \delta^{-1})$.
%
\subsection{Reduction to  the Two-Party Protocol}

As outlined in the proof of Theorem \ref{thm:noisylower}, we start by  showing that, if we have a (valid) Almost Consensus protocol, we can convert it into a two-party communication protocol between party $A$ and party $B$ with certain properties. More formally, we give the following definition.
\begin{definition}
    \label{def:2party}
    A \twopart{t}{\delta}{\epsilon}
    is a two party noisy communication protocol between parties $A$ and $B$  such that
    \begin{itemize}
        \item[i)] $B$ starts with a bit $b$ and at the end $A$ outputs one bit,
        \item[ii)] $A$ receives $t$  messages, each one of one bit,
        \item[iii)] Each  bit of communication passes through a binary symmetric channel that flips the bit with probability $\frac 12 - \epsilon$,
         \item[iv)] If $b=0$ then $A$ outputs 0 with probability $\geq 1 - O(\delta)$, 
         \item[v)] If $b=1$ then $A$ outputs 0 with probability $\leq \frac 34 +  O(\delta)$.
    \end{itemize}
\end{definition}

We can now state the formal reduction result.
\begin{lemma}
    \label{lem:reduction}
    Let $\sP$ be a protocol that solves $\delta$-Almost Consensus
    problem in $t$ rounds with probability at least $1 - \delta$ in the noisy \general\ \pull model for some $0<\delta<1/8$.
    Then, there exists a \twopart{t}{\delta}{\epsilon}, where $\epsilon$ is the noisy parameter.
\end{lemma}

\begin{proof}
    As defined in the previous section, let $\vinput{j}$ be the initial opinion vector such that the first $j$ nodes initially support opinion $0$ and the others support opinion $1$.
    Let ``$\sP \rightarrow x$'' be the event ``$\sP$ converges to a consensus where all nodes, but at most $\delta n$, output opinion $x$'' where $x \in \{ 0,1\}$. By hypotheses of the lemma we have that
    
    \begin{itemize} 
        \item During the execution of $\sP$, each node $v$ exchanges at most $t$ one-bit messages;
        
        \item $\Prob{\sP \rightarrow 0 \vert \vinput{n}} \geq 1 - \delta$; 
        \item $\Prob{\sP \rightarrow 1 \vert \vinput{0} } \geq 1 - \delta$; 
        \item For any initial opinion vector $\vinput{j}$, the probability that $\sP$ reaches an almost consensus is at least $1- \delta$, namely $\Prob{\sP \rightarrow 0 \vee \sP \rightarrow 1 \vert \vinput{j}} \geq 1 - \delta$. 
    \end{itemize}
    
    %All above references to probabilities are with respect to the randomness of the channels used by the nodes to communicate.
    Thanks to the agreement property we know that $\Probi{\sP \rightarrow 0 \vee \sP \rightarrow 1 \cond \vinput{\frac{n}{2}}} \geq 1- \delta$. Since there are only two possible opinions, then 
    $\Probi{\sP \rightarrow 0 \cond \vinput{\frac{n}{2}}} \geq \frac{1- \delta}{2}$ or $\Probi{\sP \rightarrow 1 \cond \vinput{\frac{n}{2}}} \geq \frac{1- \delta}{2}$.
    W.l.o.g.\ we assume that it holds $\Probi{\sP \rightarrow 1 \cond \vinput{\frac{n}{2}}} \geq \frac{1- \delta}{2}$, indeed if not true, we can simply rename opinion $0$ with opinion $1$.
    Now we leverage on this property in order to show that, starting from $\vinput{\frac{n}{2}}$, a large fraction of  nodes have a constant probability to output opinion $1$. Let $u$ be any node, we define ``$u \rightarrow x$'' be the event ``$u$ outputs opinion $x$ in $\sP$'', where $x \in \{0,1\}$.
    
    \begin{fact}
        \label{fact:probabilityonenode}
        If $\Probi{\sP \rightarrow 1 \cond \vinput{\frac{n}{2}}} \geq \frac{1- \delta}{2}$ then
        a node subset $S$ with $\vert S \vert \geq (1- 2\delta)n$ exists such that
        for any $u \in S$ it holds  $\Probi{u \rightarrow 1 \vert \vinput{\frac{n}{2}} } \geq \frac{1}{4} - \frac{\delta}{4}$.
    \end{fact}
    \begin{proof}
        Let $V'$ be the set of nodes $v$ such that $\Probi{v \rightarrow 0 \cond \sP \rightarrow 1, \vinput{\frac{n}{2}}} \geq \frac{1}{2}$, and let $S=V\setminus V'$. Since the expectation of the  number $Z$ of nodes that output $0$ conditioned to the event ``$\sP \rightarrow 1$'' is at most $\delta n$, the size of $V'$ is at most $2 \delta n$. Indeed,
        \begin{align*}
            \delta n & \ge \Ex{Z \cond \sP \rightarrow 1 , \vinput{\frac{n}{2}}}=\sum_{v \in V} \Prob{v \rightarrow 0 \cond \sP \rightarrow 1, \vinput{\frac{n}{2}}}   
            \\ 
            & \ge \sum_{v \in V'} \Prob{v \rightarrow 0 \cond \sP \rightarrow 1, \vinput{\frac{n}{2}}} \ge \frac 12 |V'|. 
            \end{align*}
            \noindent This implies that $\vert S \vert \geq (1- 2\delta)n$. To conclude the proof, we observe that,    for any $u \in S $, we have
            \begin{align*}
            \Prob{u \rightarrow 1 \cond \vinput{\frac{n}{2}}} 
            &\geq \Prob{u \rightarrow 1 \wedge \sP \rightarrow 1 \cond \vinput{\frac{n}{2}}}\\
            &= \Prob{u \rightarrow 1 \cond \sP \rightarrow 1, \vinput{\frac{n}{2}}} \cdot \Prob{\sP \rightarrow 1 \cond \vinput{\frac{n}{2}}}\\
            &\geq \frac{1}{2} \cdot \left(\frac{1}{2} - \frac{\delta}{2} \right) = \frac{1}{4} - \frac{\delta}{4}.
        \end{align*}
    \end{proof}
    
    We now consider the initial opinion vector $\vinput{n}$, where all the nodes support   opinion $0$. Remind that $\Prob{\sP \rightarrow 0 \cond \vinput{n}} \geq 1 - \delta$. Using a similar argument of Fact~\ref{fact:probabilityonenode}, we  can prove the following
    
    \begin{fact}
        \label{fact:probabilityonenode2}
        If $\Prob{\sP \rightarrow 0 \cond \vinput{n}} \geq 1 - \delta$ then
        a node subset $H$ with $\vert H \vert \geq \frac{3}{4}n$ exists such that
        for any $u \in H$, $\Prob{u \rightarrow 0 \cond \vinput{n} } \geq 1 - 5\delta$.
    \end{fact}
    \begin{proof}
        Let $V'$ be the set of nodes $v$ such that $\Prob{v \rightarrow 1 \cond \sP \rightarrow 0, \vinput{n} } \geq 4 \delta$, and set $H=V\setminus V'$. Since the expected number of nodes that output opinion $1$ conditioned to the event "$\sP \rightarrow 0$" is at most $\delta n$, the size of $V'$ is at most $n/4$, and hence $|H| \ge \frac{3}{4} n$.
        
        To conclude the proof, observe that, for any $u \in H$,
        \begin{align*}
            \Prob{u \rightarrow 0 \vert \vinput{n}} 
            &\geq \Prob{u \rightarrow 0 \wedge \sP \rightarrow 0 \vert \vinput{n}}\\
            &= \Prob{u \rightarrow 0 \vert \sP \rightarrow 0, \vinput{n}} \cdot \Prob{\sP \rightarrow 0 \vert \vinput{n}}\\
            &\geq (1 - 4\delta) \cdot (1 - \delta) \geq 1 - 5\delta.
        \end{align*}
    \end{proof}
    
    By combining Facts \ref{fact:probabilityonenode} and   \ref{fact:probabilityonenode2}, we obtain the following 
    \begin{fact}
        \label{fact:thereisamagicnode}
        Let $0<\delta<1/8$. At least one node $v^*$ exists 
        such that: \\
        (i) its initial opinion is $0$ in both  $\vinput{\frac{n}{2}}$ and $\vinput{n}$, and \\ 
        (ii)  
        $\Prob{v^* \rightarrow 1 \vert \vinput{\frac{n}{2}} } \geq \frac{1}{4} - \frac{\delta}{4}$   and   
        $\Prob{v^* \rightarrow 1 \vert \vinput{n}} \leq 5\delta $. 
    \end{fact}
    \begin{proof}
        From Facts \ref{fact:probabilityonenode} and \ref{fact:probabilityonenode2}, if $\delta < \frac{1}{8}$ then   $\vert S \vert > \frac{3}{4}n$, and thus $\vert S \cap H \vert > \frac{1}{2}n$. Since the nodes   having    initial opinion $0$ in both   vectors $\vinput{\frac{n}{2}}$ and $\vinput{n}$ are exactly $\frac{n}{2}$,  it must exist    at least one node having  Properties (i) and (ii).    
    \end{proof}
    
    % Now we also consider the event $\zeta$ = "$v^*$ receives at most $t$ bits of messages". Since this event has probability at least $1 - \delta$ it holds
    % $\Prob{v^* \rightarrow 1, \zeta \vert \vinput{\frac{n}{2}} } \geq \frac{1}{4} - \frac{5}{4}\delta$ and $\Prob{v^* \rightarrow \text{zero}, \zeta \vert \vinput{n}} \geq 1 - \Prob{v^* \rightarrow 1 \vert \vinput{n}} - \Prob{\zeta \vert \vinput{n}} \geq 1-  6\delta$.\\
    
    We now realize a \twopart{t}{\delta}{\epsilon} between parties $A$ and $B$. If $B$ has input $0$, then $A$ simulates $v^*$ with initial opinion $0$ and $B$ simulates all other nodes as if they had all initial opinion $0$. If $B$ has input $1$, then $A$ simulates $v^*$ with initial opinion $0$ and $B$ simulates all other nodes as if they had as an initial opinion vector of $\frac{n}{2}$ ones and $\frac{n}{2}-1$ zeroes. In the simulation, $A$ and $B$ need to communicate (via the binary symmetric channel) only when $v^*$ sends or receives messages in $\sP$. At the end, $A$ will output $0$ or $1$, depending on the outcome of the Almost Consensus protocol $\sP$.
    % Finally, if $A$ receives more than $t$ bits of messages during the execution, the protocol fails.\\ 
    %
    
    Note that at each round of simulation of $\sP$, $v^*$ receives exactly\footnote{At the beginning of the proof we assumed w.l.o.g.\ that at each round each node pulls another node.} $1$ bit and no other information is available to it. 
    %
    %%Ho commentato tutto il discorso sull'identità del nodo dal quale si pulla: secondo me crea troppa confusione. Nel modello \pull classico non si assume che ci sono ID, inoltre con "identità" bisognerebbe appunto chiarire se si intende che il nodo che pulla sa su che porta sta pullando, o due nodi che pullano sullo stesso nodo vedono lo stesso id... insomma meglio spostare in un remark successivo alla prova
    % Indeed, in the \pull model, the only metadata\footnote{Inspired by the use of the term in software engeneering, we call \emph{metadata} any type of data about a message other than the content of the message itself.} attached to the message is the identity of the sender (i.e., the identity of  the pulled node). Moreover, notice that in the (deterministic) \pull model the sender is decided by $v^{*}$, while the uniform \pull model the sender is decided by an external source of randomness. In both cases, the identity of the sender can not be used to convey information to $v^*$ from the rest of the graph. 
    %
    Hence, in the resulting Two-Party Protocol, the only information obtained by $A$ is the bit received in that round from $B$ (corresponding to the rest of the graph).    
    
    Thanks to Fact \ref{fact:thereisamagicnode},
    for any $\delta < \frac{1}{8}$, if $B$ has input $0$, with probability at least $1 - 5\delta$, $A$ will output $0$. On the other hand, if $B$ has input $1$, with probability at least $\frac{1}{4} - \frac{\delta}{4}$, $A$ will output $1$. 
\end{proof}

\subsection{Lower bound for the Two-Party Protocol}
\label{ssec:lbTwoParty}

\begin{lemma}
    \label{lem:kl}
    Any  \twopart{t}{\delta}{\epsilon} requires a number of rounds
    $t$ such that $t = \Omega(\epsilon^{-2} \log \delta^{-1})$.
\end{lemma}

\begin{proof}
    In     any    interaction between $A$ and $B$ of $t$ rounds,  we name   {\em view}   the sequence $\bw$   of all $1$-bit messages received by $A$ during the interaction. Notice that this sequence determines the sequence of messages sent by $A$ and the final output of $A$.
    Let $\bX$ be the random variable that represents the (random) view when $B$'s input is 1 and   let  $\bY$ be the random variable that represents  the (random) view  when $B$'s input is 0. 
    
    %\andy{$a$, $X$ e $Y$ sono vettori di bit di lunghezza $t$: in %accordo con il resto del paper, li dovremmo mettere in bold.}
     
    Recall that the Kullback–Leibler divergence between $\bX$ and $\bY$ is defined as
    
    \[ 
        D(\bX || \bY) := \sum_{\bw\in \{0,1\}^t} \Prob{\bX=\bw}  \log \frac{\Prob{\bX=\bw } }{\Prob{\bY= \bw} },
    \]
    
     We will prove the following two facts which easily  imply the claim of the lemma:
    
    \begin{equation} 
        \label{kl.lower}
        D(\bX || \bY) \geq \Omega \left( \log \frac 1 \delta \right),
    \end{equation}
    
    \begin{equation} 
        \label{kl.upper}
        D(\bX || \bY) \leq O  \left(  t {\epsilon^2} \right ).
    \end{equation}
    
    To prove (\ref{kl.lower}) we use the {\em data-processing inequality} which (in particular) states that for every (possibly random) function $f(\cdot)$ we have
    \begin{equation}
        \label{eq:dataproc}
        D(\bX || \bY)  \geq D(f(\bX) || f(\bY)).
    \end{equation}
    
    For a view $\bw$, define $f(\bw)$ to be the output of $A$ for that view. Then $f(\bX)$ and $f(\bY)$ are 0/1 random variables. If
    we call $p:= \Prob{f(\bX) = 1}$ and $q:= \Prob{f(\bY)=1}$ we get
    \[ 
        D(f(\bX) || f(\bY))  =  p \log \frac pq + (1-p) \log \frac {1-p}{1-q},
    \]
    and recalling that $p \geq \Omega(1)$ and $q \leq O(\delta)$ we get
    \begin{equation}
        \label{eq:lowerOnKL}
        D(f(\bX) || f(\bY) ) \geq \Omega \left( \log \frac 1 \delta \right),
    \end{equation}
    which proves (\ref{kl.lower}).
    
    To prove (\ref{kl.upper}) we use  the {\em chain rule}. If we have two pairs of  jointly distributed random variables $(X,X')$ and $(Y,Y')$,
    then the conditional KL divergence is defined as
    \[ 
        D( X'|X \ || \ Y'|Y ) = \E_{x\sim X, y\sim Y} D( (X' | X=x)  \ || \ (Y' | Y=y) ).
    \]
    Let $\oplus$ be the operation that denotes the concatenation between two sequence of random variables, the chain rule is
    \[ D (X \oplus X' || YY') = D(X||Y) + D( (X'|X)  \ || \ (Y'|Y) ) \]
    
    Since any view has length $t$, we  write
    $\bX= X_1 \oplus \ldots \oplus X_t$ and $\bY=Y_1 \oplus \ldots \oplus Y_t$. Then we can write the KL divergence of $\bX$ and $\bY$ as
    \[ 
        D(\bX||\bY) = D(X_1 || Y_1 ) + D((X_2 | X_1) \ || \ (Y_2 | Y_1) ) + \cdots + D((X_t | X_1\cdots X_{t-1} ) \ || \ (Y_t | Y_1\cdots Y_{t-1} ) ).
    \]
    
    For each $i$ we can easily compute $D((X_i | X_1\cdots X_{i-1} ) \ || \ (Y_i | Y_1\cdots Y_{i-1} ) )$, indeed both ``$ X_i | X_1\ldots X_{i-1} = w_1\ldots w_{i-1} $''
    and
    ``$ Y_i | Y_1\ldots Y_{i-1} = w'_1\ldots w'_{i-1} $''
    are binary random variables such that, for any $w = \{0,1\}$, it holds
    
    \begin{equation}\label{eq:kld}
            \frac{1}{2} - \epsilon \leq \Prob{X_i = w | X_1\ldots X_{i-1}} , \Prob{Y_i = w | Y_1\ldots Y_{i-1}}\leq \frac{1}{2} + \epsilon.
    \end{equation} 
    Thus we have
    \begin{align}
        &D((X_i | X_1\cdots X_{i-1} ) \ || \ (Y_i | Y_1\cdots Y_{i-1} ) )\nonumber \\
        &= \Prob{ X_i = 0 | X_1\cdots X_{i-1}} \log \frac{\Prob{X_i = 0 | X_1\cdots X_{i-1}}}{\Prob{Y_i = 0 | Y_1\cdots Y_{i-1}}} \\
        &\qquad+ \Prob{ X_i = 1 | X_1\cdots X_{i-1}} \log \frac{\Prob{X_i = 1 | X_1\cdots X_{i-1}}}{\Prob{Y_i = 1 | Y_1\cdots Y_{i-1}}}\nonumber\\
        &< \left| \log \frac{\Prob{X_i = 0 | X_1\cdots X_{i-1}}}{\Prob{Y_i = 0 | Y_1\cdots Y_{i-1}}} \right| + \left| \log \frac{\Prob{X_i = 1 | X_1\cdots X_{i-1}}}{\Prob{Y_i = 1 | Y_1\cdots Y_{i-1}}} \right|\nonumber \\
        &\leq 2 \left| \log \frac{\frac{1}{2} + \epsilon}{\frac{1}{2} - \epsilon} \right|
        = 2 \left| \log \frac{1 + 2\epsilon}{1 - 2\epsilon} \right|\\
        &= 2 \left| \log(1 + 2\epsilon) - \log( 1 - 2\epsilon) \right|\\
        &\stackrel{(a)}{=} 2 \left| \epsilon^2 + o(\epsilon^2) \right| = O(\epsilon^2),\nonumber
    \end{align}
    where in (a) we use the Taylor approximation. 
    Thus, we can conclude that
    \[ 
        D(\bX||\bY) = \sum_{i=1}^t D((X_i | X_1\cdots X_{i-1} ) \ || \ (Y_i | Y_1\cdots Y_{i-1} ) ) \leq O(t\epsilon^2).
    \]
\end{proof}

%\subsection{Generality of Theorem \ref{thm:noisylower}}\label{ssec:lbgeneral}

%In this section we briefly discuss some limitations and extension of Theorem \ref{thm:noisylower}. 

% IT SEEMS WE PROPOSE TO REMOVE THE FOLLWING SUBS.

\subsection{Absence of reliable components in communication mechanism}\label{ssec:lbgeneral}

As discussed in the Introduction, Theorem \ref{thm:noisylower} should be contrasted with the result on the noisy uniform \push model \cite{FHK17,fraigniaud_noisy_2018}, in which at each round a node may send a bit to a random neighbor and, upon being received, the bit may be flipped by the communication noise with probability $\frac 12 - \epsilon$. 

In \cite{FHK17} it is assumed that the protocol satisfies a \emph{symmetry} hypothesis, i.e. the choice of nodes on whether to communicate or not, cannot depend on the value of the bit that the node wish to communicate. Without such assumption, the action of communicating \textit{anything} can be employed to reliably solve the valid consensus problem (and many others).

More precisely, we can have nodes sending messages at even rounds\footnote{Note that this expedient relies on nodes sharing a synchronous binary clock. In \cite{FHK17}, it is shown how such clock can be easily obtained if, for example, nodes are initially inactive and become active upon receiving the first message.} only if they wish to communicate value $0$, and at odd rounds only if they wish to communicate value $1$. 
%Thus, by the mere fact that a node receives any message at an even/odd round, it can reliably infer the value of the original message without being affected by noise.

Our lower bound is indeed not applicable to the noisy uniform \push model, because it is not possible to reduce Consensus to the Two-Party Protocol in this model. Precisely, Definition~\ref{def:2party} requires that each bit of information passes through the noisy channel (property \textit{iii)}). 
In the noisy general \pull model, this is verified at the end of the proof of Lemma \ref{lem:reduction}. 
However, this is not the case in the noisy uniform \push model since, besides the (noisy) content of the message, the message received by $v^{*}$ communicates to it the fact that another node \emph{has chosen to communicate something} at the present round.

Hence, the above comparison with the noisy \push model essentially suggests that the lower bound in Theorem \ref{thm:noisylower} sensibly relies on the fact that no component of the communication model is immune to the noise action.

\subsection{A lower bound for Broadcast}
As discussed in Section~\ref{ssec:noisy}, \cite{BFKN18} gives an $\Omega(\epsilon^{-2} n)$ lower bound for Broadcast in the noisy uniform \pull model. 
Using a similar argument to that used in the proof of our Lemma \ref{lem:kl} in Section \ref{ssec:lbTwoParty}, in what follows we give a sketch of  how  the lower bound above can be strengthen to  $\Omega( \epsilon^{-2} n \log n)$ that holds for any protocol solving the binary Broadcast  problem  w.h.p. 

\begin{lemma}
    \label{lem:broadcastlower}
    Any protocol that solves  Broadcast     in the noisy uniform \pull model w.h.p.  requires 
    $\Omega( \epsilon^{-2} n \log n)$ rounds. 
\end{lemma}
\begin{proof}[Sketch of Proof.]
In \cite{BFKN18}, the authors prove the following: let $\sP$ be a $t$-round protocol for Broadcast in 1-bit noisy uniform \pull   with error parameter $\epsilon$ starting from an arbitrary source.  Let $D_0$ and $D_1$ be the distributions of the sequence of  all received messages by all nodes through all the  $t$ rounds, assuming that the source message is $0$ or $1$, respectively. Then the KL divergence between $D_0$ and $D_1$ is $O(t \epsilon^2 / n)$. If the protocol succeeds with constant probability, then the statistical distance between $D_0$ and $D_1$ has to be $   \Omega(1)$ and the KL divergence between $D_0$ and $D_1$ also has to be $\Omega(1)$, leading to the $t = \Omega(\epsilon^{-2}n)$ lower bound of \cite{BFKN18}. 

The key observation here is that if the success probability is set to be no smaller than $1-\delta$, for some $0<\delta<1$, then, by using a similar argument to that in Lemma \ref{lem:kl}, we derive that the KL divergence $D_0$ and $D_1$ must be $\Omega(\log \delta^{-1})$, and so the lower bound is $\Omega(\epsilon^{-2} n \log \delta^{-1})$.

\end{proof}

\section{Upper Bounds in the Noisy Model}\label{sec:noisyub}
In Theorem~\ref{thm:noisylower}, we obtained a lower bound
$\Omega(\epsilon^{-2} \log n)$ on the number of rounds required by any protocol
for Almost Consensus and Almost Majority Consensus that works w.h.p., in the
noisy \general\ \pull model. In the next section we show that, in this model,
that lower bound is tight for both tasks. As discussed in the introduction,
combined with the lower bound for Broadcast in~\cite{BFKN18} this result
demonstrates a strong complexity gap between Consensus and Broadcast in the
noisy uniform \pull.

\subsection{Upper bound for Consensus in noisy \pull}\label{ssec:ubNoisyCons}

\begin{theorem}\label{thm:upperbound}
In the noisy uniform \pull model, with noisy parameter $\epsilon$, a protocol
exists that achieves Consensus within $O(\epsilon^{-2} \log n)$ rounds and
communication, w.h.p. The protocol requires $\Theta(\log\log n + \log
\epsilon^{-2})$ local memory. 

Moreover, if the protocol starts from any initial opinion vector with bias
$\bias = \Omega(\sqrt{n\log n})$, then it guarantees Majority Consensus, w.h.p.
\end{theorem}

\noindent
The protocol we refer to in the above theorem works in two consecutive phases.
Each phase is a simple application of the well-known $k$-\emph{Majority
Dynamics}~\cite{becchetti2016stabilizing,BCNPST14}:

\begin{quote}
\textsc{$k$-Majority}. \emph{At every round, each node samples $k$
neighbours\footnote{In the binary case when $k$ is odd, the $k$-Majority is
stochastically equivalent to the $k+1$-Majority where ties are broken u.a.r.\
(see Lemma~17 in~\cite{fraigniaud_noisy_2018}). For this reason, in this
section we assume that $k$ is odd.} independently and u.a.r. (with
replacement). Then, the node updates its opinion according to the majority
opinion in the sample.}
\end{quote}

\noindent Notice that $k$-Majority, as stated above, assumes a \emph{uniform
$k$-\pull model} where, at each round, every node can pull one message from
each of the $k$ neighbors chosen independently and uniformly at random
\emph{with replacement}\footnote{The assumption that neighbors are chosen
independently and uniformly at random with replacement is consistent with
previous work~\cite{becchetti_plurality_2015, ghaffari_polylogarithmic_2016}.}.
However, it is easy to verify that this parallel process can be implemented on
the uniform $1$-\pull\ model using additional $\Theta(k)$ local memory and with
a slowdown factor $k$ for its convergence time. In the rest of this section, we
will thus consider the following two-phase protocol on the uniform $k$-\pull
model.  

\begin{quote}
\textsc{Majority Protocol}. Let $\alpha$ be a  sufficiently large positive
constant\footnote{The value of $\alpha$ will be fixed later in the analysis.}.
Every node performs $\alpha \log n$ rounds of $k$-Majority with $k = \Theta(1 /
\epsilon^2)$, followed by one round of the $k$-Majority with $k =
\Theta(\epsilon^{-2} \log n)$.
\end{quote}

\noindent The proof of Theorem~\ref{thm:upperbound} will proceed according to
the following scheme: we will show that

\begin{itemize}
\item If $k = \Omega(1 / \epsilon^2)$ then
\begin{itemize}
\item starting from any opinion vector, within $O(\log n)$ rounds of
$k$-Majority the process reaches an opinion vector where the bias is $\bias =
\Omega(\sqrt{n} \log n)$, w.h.p.  (Lemma~\ref{thm:breaking}),
\item starting from an opinion vector with bias $\bias = \Omega(\sqrt{n \log
n})$ then, within $O(\log n)$ rounds of $k$-Majority the process reaches an
opinion vector where the bias is $\bias = \Theta(n)$ and the majority opinion
is preserved (Lemma~\ref{thm:toalinearbias}),
\end{itemize}
\item If $k = \Omega(\epsilon^{-2} \log n)$ and the opinion vector has bias
$\bias = \Theta(n)$, then in one round of $k$-Majority the process reaches
consensus on the majority opinion, w.h.p. (Lemma~\ref{thm:convergence}).
\end{itemize}

\subsection{Proof of Theorem~\ref{thm:upperbound}} 
Let $\mathbf{C}^{(t)}$ be the random variable indicating the opinion vector 
at round $t$ of the majority protocol. Let us name $R^{(t)}$ the number of nodes
supporting opinion $0$ in such opinion vector and let us define the bias at
round $t$ as $\bias^{(t)} = R^{(t)} - \left( n - R^{(t)}\right) = 2 R^{(t)} -
n$. In the rest of the section we assume w.l.o.g. that the bias is positive.

Since every time a node $u$ pulls the opinion of a node $v$, node $u$ correctly
gets the opinion of node $v$ with probability $\frac{1}{2} + \epsilon$ and it
gets the other opinion with probability $\frac{1}{2} - \epsilon$, we can 
write the probability $p_{r}$ that a node observes opinion $0$ after a pull (i.e.
after the node has sampled a neighbor and the noise has possibly flipped its
opinion), as a function of the bias 
\[
p_{r} = \left(\frac{1}{2} + \frac{\bias^{(t)}}{2n}\right)\left(\frac{1}{2} 
+ \epsilon \right) + \left(\frac{1}{2} - \frac{\bias^{(t)}}{2n}\right) 
\left(\frac{1}{2} - \epsilon \right) 
= \frac{1}{2} + \frac{\epsilon \bias^{(t)}}{n}.
\]
%
% Analogously we define $p_{\blue} = 1-p_r = \frac{1}{2} - \frac{\epsilon
% \bias^{(t)}}{n}$ as the probability to get the opinion $1$. These equations
% clearly show how the noise affects the Majority dynamics: At each round it
% behaves as in a noiseless model, but with a smaller bias. For this reason we
% compensate by choosing a sample $k$ as a function of the noise parameter.
% %
% Since we assumed that $k$ is odd, the probability that a node supports the
% opinion $0$ after a round of the $k$-Majority dynamics is thus
% \[
% \Prob{u \text{ supports 0}} 
% = \sum_{i = \left\lceil \frac{\ell }{2}\right\rceil}^k \binom{k}{i} p_r^i p_\blue^{k-i}\,.
% \]

From now on with ``$u$ supports 0''we mean that after a round of the $k$-Majority 
the node $u$ has opinion $0$, namely the most frequent opinion in the random sample is $0$. 
From Lemma~2 in~\cite{FHK17} it follows that, if each node samples 
$\Omega(1/\epsilon^2)$ neighbors, then the bias $\bias$ grows exponentially, in 
expectation, until it reaches linear size. 

\begin{lemma}[Lemma 2 in \cite{FHK17}]\label{lem:bin_maj_ampl}
Let $k = c/\epsilon^2$ be an odd integer for a sufficiently large $c$, and
$\mathbf{c}$ be any opinion vector with bias $\bias$, then it holds that
\[
\Prob{u \emph{ supports 0} \,\vert\, \mathbf{C}^{(t)} 
= \mathbf{c}} \geq \frac{1}{2} + \min\left\{4\bias \,,\, \frac{1}{100}\right\}.
\]
\end{lemma}

The above lemma is useful to give high probability results on the behaviour of
the protocol when the bias is large enough to guarantee concentration around
the expectation. It thus remains to handle the cases in which the bias is so
small that its expected multiplicative growth is smaller than its standard
deviation. By leveraging on Lemma~4.5 in~\cite{clementi_tight_2017} (see
Appendix~\ref{lemma4.5}), the next lemma shows how the variance of the process
and the multiplicative drift of the bias break symmetry, when the initial
opinion vector has small or no bias. 

\begin{lemma}\label{thm:breaking}
Let $m$ be any positive constant. Starting from an opinion vector with bias $0
\leqslant \bias^{(t)} < m \sqrt{n} \log n$, the system reaches an opinion
vector with bias $\bias^{(t)} \geq m \sqrt{n} \log n$ within $O(\log n)$ rounds
of the $k$-Majority dynamics, w.h.p.
\end{lemma}

\begin{proof}
Let us name $\Omega$ the space of all opinion vectors. In order to apply 
Lemma~4.5 in~\cite{clementi_tight_2017} we need to show that the following two
properties hold:

\begin{itemize}
\item For any positive constant $h$, a positive constant $c_1 < 1$ exists
such that, for any opinion vector $\mathbf{c} \in \Omega$ with $\bias^{(t)} < m
\sqrt{n} \log n$, it holds that
\[
\Prob{\bias^{(t+1)} < h\sqrt{n} \cond \mathbf{C}^{(t)}  = \mathbf{c}} < c_1;
\]
\item Two positive constants $\epsilon'$ and $c_2$ exist such that, for
every opinion vector $\mathbf{c} \in \Omega$ with $\bias^{(t)} < m \sqrt{n} 
\log n$, it holds that
\[
\Prob{\bias^{(t+1)} < (1+\epsilon')\bias \cond \mathbf{C}^{(t)} = \mathbf{c}} 
< e^{-c_2 \bias^2/n}.
\]
\end{itemize}

Let $\mathbf{c} \in \Omega$ such that $\bias^{(t)} < m
\sqrt{n} \log n$. As for the first property, since $\bias^{(t)} = 2 R^{(t)} - n$ then
$\Var{\bias^{(t+1)}} = \Theta(\Var{R^{(t+1)}}) = \Theta(n (\frac{1}{2} +
2\frac{\bias^{(t)}}{n})(\frac{1}{2} -  2\frac{\bias^{(t)}}{n}))$. Notice that
in this symmetry-breaking phase $\frac{\bias^{(t)}}{n} = o(1)$, thus
$\Var{\bias^{(t+1)} \cond \mathbf{C}^{(t)} = \mathbf{c} } = \Theta(n)$. Hence, a simple application of the
Berry-Esseen Theorem (Theorem~\ref{thm:berreyesseen} in
Appendix~\ref{berry-esseen}) shows that, for a sufficiently large $n$, there is
constant probability that the bias becomes $\Theta(\sqrt{n})$.

As for the second property, for any $\mathbf{c} \in \Omega$ with $\bias^{(t)} < m \sqrt{n} 
\log n$, by Lemma~\ref{lem:bin_maj_ampl} we have that
\[
\Ex{R^{(t+1)} \,\vert\, \mathbf{C}^{(t)} 
= \mathbf{c} } 
= \min\left( n \left( \frac{1}{2} +  4\frac{\bias^{(t)}}{n} \right) , n
\left( \frac{1}{2} +  \frac{1}{100} \right) \right) 
= n \left( \frac{1}{2} + 4\frac{\bias^{(t)}}{n} \right)\,.
\]
By applying the additive form of the Chernoff Bound (see 
Appendix~\ref{apx:cbaf}) with $\lambda = 2\bias^{(t)}$ we get
\[
\Prob{R^{(t+1)} < n \left(\frac{1}{2} +  4\frac{\bias}{n}\right) - 2\bias \cond \mathbf{C}^{(t)} = \mathbf{c} } 
< e^{-8 \bias^2 / n}.
\]
Thus, with probability at least $1 - e^{-8 \bias^2 / n}$ it holds 
\[
\bias^{(t+1)} = 2 R^{(t+1)} - n
\geq  2n\left(\frac{1}{2} +  4\frac{\bias}{n}\right) - 2\bias - n
= 2 \bias\,.
\]
\end{proof}

Once the bias is large enough, by iteratively using
Lemma~\ref{lem:bin_maj_ampl} we can also show that the bias reaches $\Theta(n)$
within $O(\log n)$ rounds and that the majority opinion is preserved, w.h.p.

\begin{lemma}\label{thm:toalinearbias}
Let $m$ be any positive constant. Starting from an opinion vector with bias
$\bias^{(t)} \geq m \sqrt{n \log n}$, the system reaches an opinion vector with
bias $\bias^{(t)} = \Theta(n)$ within $O(\log n)$ rounds of the $k$-Majority
dynamics, w.h.p., and the sign of the bias is preserved, w.h.p.
\end{lemma}

\begin{proof}
As long as $4\bias \leq \frac{n}{100}$, similarly to the second point of
Lemma~\ref{thm:breaking}, from Lemma~\ref{lem:bin_maj_ampl} it follows that 
\[
\Prob{\bias^{(t+1)} < 2 \bias \cond \mathbf{C}^{(t)}
= \mathbf{c}} < e^{-8\bias^2/n}.
\]
Using the fact that $\bias^{(t)} \geq m \sqrt{n \log n}$ we obtain that
\[
\Prob{\bias^{(t+1)} < 2 \bias \cond \mathbf{C}^{(t)}
= \mathbf{c}} < e^{8 m \log n} < \frac{1}{n^{8 m}} 
= \frac{1}{n^{\Theta(1)}},
\]
Thus, the bias grows of a multiplicative factor at each round, w.h.p., and this
implies, by using the Union Bound, that it reaches a value
$\bias^{(t)} \geq \frac{n}{400}$ within $O(\log n)$ rounds, w.h.p.
\end{proof}

For many values of $\epsilon$, the use of the $k$-Majority dynamics with $k =
\Theta(1/\epsilon^2)$ is not sufficient to reach an opinion vector where all, or
at least many, nodes support the same opinion. Indeed, if $\epsilon$ is a
constant smaller than $\frac{1}{2}$ and we consider an opinion vector $\mathbf{c}_0$ where all
the nodes already support opinion $0$, it is easy to see that after one round a
constant fraction of the nodes will change opinion, w.h.p.: let $X_u$ be the 
random variable counting the number of times a node $u$ pulls opinion $1$ in one
round starting from such opinion vector. It holds that 
\[
\Ex{X_u \cond \mathbf{C}^{(t)}
= \mathbf{c}_0} = \frac{c}{\epsilon^2}\left(\frac{1}{2} - \epsilon \right) 
= \frac{c}{2\epsilon^2} - \Theta\left(\frac{1}{\epsilon}\right) 
\]
The variance of $X_u$ is $\frac{c}{\epsilon^2}\left(\frac{1}{2} - \epsilon
\right)\left(\frac{1}{2} + \epsilon \right) = \Theta\left(1 /
\epsilon^2\right)$ and thus the standard deviation is $\Theta\left(1/
\epsilon\right)$. It follows from the Berry-Esseen Theorem (see
Appendix~\ref{berry-esseen}) that $X_u$ has constant probability to deviate
from his expectation of a quantity $\Theta\left(1/ \epsilon\right)$, that is
sufficient to increase $X_u$ above $\frac{c}{2\epsilon^2}$, making opinion $1$
the majority in the sample. Hence, each node has constant probability to change
opinion and by a simple application of a Chernoff Bound and the Union Bound it
holds that at least a constant fraction of nodes will change opinion with
probability $1 - e^{-\Theta(n)}$.

The above example shows that, in order to reach consensus, we need to increase
the size of the sample of each node. In the next lemma we show, as an
application of a Chernoff Bound and the Union Bound, that considering a
sufficiently large sample, in one step the process reaches an opinion vector
where only one opinion is present.

\begin{lemma}\label{thm:convergence}
For any positive constant $c_3$, a suitable constant $c_4$ exists such
that, starting from an opinion vector with bias $\bias^{(t)} \geq \frac{c_3
n}{2}$, in one round of the $k$-Majority dynamics with $k = c_4 \epsilon^{-2}
\log n$, the system reaches an opinion vector with $\bias = n$,
w.h.p.
\end{lemma}

\begin{proof}
Let $\mathbf{c}$ a opinion vector such that $\bias^{(t)} \geq \frac{c_3
n}{2}$. The proof is a straightforward application of the multiplicative form of the
Chernoff Bound (see Appendix~\ref{apx:cbmf}) and of the Union Bound. Indeed,
let $X_u$ be the random variable counting the number of times a node $u$ pulls
opinion $1$ in one round starting from an opinion vector with bias $\bias^{(t)}
\geq \frac{c_3 n}{2}$. It holds that
\begin{align*}
\Ex{X_u \cond \mathbf{C}^{(t)}
= \mathbf{c} } & = \frac{c_4 \log n}{\epsilon^2} 
\left(
\left( \frac{1}{2} + c_3 \right) 
\left( \frac{1}{2} + \epsilon \right) 
+ \left( \frac{1}{2} - c_3 \right) 
\left( \frac{1}{2} - \epsilon \right)
\right)\\ 
& = \frac{c_4 \log n}{\epsilon^2} \left(\frac{1}{2} + 2 \epsilon c_3 \right)
\end{align*}
Notice that, if $X_u \geq \frac{c_4 \log n}{\epsilon^2} \left(\frac{1}{2} +
\epsilon c_3 \right)$, then the majority opinion is $1$. The probability of the
complementary event is
\begin{align*}
\Prob{X_u < \frac{c_4 \log n}{\epsilon^2} \left(\frac{1}{2} + \epsilon c_3 \right) \cond \mathbf{C}^{(t)}
= \mathbf{c}}
&= \Prob{X_u < \frac{c_4 \log n}{\epsilon^2} \left(\frac{1}{2} + 2 \epsilon c_3 \right) - \frac{c_3 c_4 \log n}{\epsilon} \cond \mathbf{C}^{(t)}
= \mathbf{c}}\\
&= \Prob{X_u < \Ex{X_u}\left(1 - \frac{c_3 \epsilon}{\frac{1}{2} + 2\epsilon c_3}\right) \cond \mathbf{C}^{(t)}
= \mathbf{c}} \\
&< \exp \left(- \frac{c_4 \log n}{\epsilon^2} \left(\frac{1}{2} + 2\epsilon c_3 \right) \left(\frac{c_3 \epsilon}{\frac{1}{2} + 2\epsilon c_3} \right)^2 \right)\\
&= \exp \left(- \frac{c_3^2 c_4 \log n}{\frac{1}{2} + 2\epsilon c_3} \right)\\
&\leq \exp \left(- \frac{c_3^2 c_4 \log n}{\frac{1}{2} + 2} \right)
= \frac{1}{n^{2}}
\end{align*}
where in the last equality we chose $c_4 = \frac{4}{5} c_3^2$. The thesis then
follows by applying the Union Bound over the $n$ nodes.
\end{proof}

\subsection{A tight bound for Broadcast in noisy uniform \pull}\label{ssec:ubNoisyBroad}
We show an $\bigO(\epsilon^{-2} n \log n)$ upper bound for Broadcast in the noisy uniform \pull model.
This is obtained via a  simple protocol \textsc{NoisyBroadcast}  we describe below.
Observe that the protocol assumes that nodes know the value of the noise parameter $\epsilon$. 

\smallskip
\textbf{Protocol \textsc{NoisyBroadcast}.} 
\begin{itemize}
    \item In the first phase, each non-source node displays $0$ (obviously, the source  displays its input value), and  performs a pull operation for $\Theta(\epsilon^{-2} n \log n) $ rounds; it then chooses to support value $1$ iff the fraction of received messages equal to 1 is at least $\frac 12 - \epsilon(1-\frac{1}{2n})$, zero otherwise. 
    
    \item In the second phase, nodes run the Majority Consensus protocol of Theorem \ref{thm:upperbound}, starting with the value obtained at the end of the first phase. 
\end{itemize}
 
We prove the following theorem. Notice that this upper bound is tight, since it matches the lower bound in Lemma~\ref{lem:broadcastlower}.

\begin{theorem}
    \label{thm:broadcastupper}
    Protocol \textsc{NoisyBroadcast} solves the Broadcast problem in the noisy uniform \pull model in $\bigO (\epsilon^{-2} n \log n)$ rounds, w.h.p. 
\end{theorem}
\begin{proof}
    We prove that at the end of the first phase, the fraction of nodes which have obtained a value equal to the source's input is greater than those that failed by at least $\sqrt{ n \log n }$ nodes.
    The latter fact satisfies the hypothesis of Theorem \ref{thm:upperbound} for solving Majority consensus in $\bigO (\epsilon^{-2} \log n)$, which constitutes the second phase. 
    
    Recall that, from the definition of \textsc{NoisyBroadcast}, during the first phase all (non-source) nodes displays value 0. 
    Hence, if the source's input value is also 0, each pulled message is equal to 1 with probability $\frac 12 - \epsilon$. 
    Otherwise, if the source's input value is $1$, each pulled message is equal to 1 with probability\footnote{We remark that, in order to simplify calculations, we assume that each node may pull itself with probability $\frac 1n$.} $\frac 12 - \epsilon(1-\frac 1n) + \frac {\epsilon}{n}$. 
    
    Let us call the two aforementioned configurations $\conf 0$ and $\conf 1$, respectively. 
    
    By the Berry-Esseen Theorem (Theorem \ref{thm:berreyesseen}), the distribution of the $n\log n$ pulled messages in $\conf 0$ and $\conf 1$ is given (up to a $\bigO (\frac 1{\sqrt{n}})$ error in total variation), by the following two normal distribution: the first, with expectation $(\frac 12 - \epsilon) \frac {n \log n}{\epsilon^2} $ and variance less than $\frac 1{2\epsilon} \sqrt{ n \log n}$; the second, with expectation $(\frac 12 - \epsilon(1-\frac 1n) + \frac {\epsilon}{n}) \frac {n \log n}{\epsilon^2} $ and variance less than $\frac 1{2\epsilon}\sqrt{ n \log n}$. 
    
    In order to choose value $1$ at the end of the phase, the rule of the protocol requires that, in $\conf 0$, the number of received one exceeds the expectation $(\frac 12 - \epsilon) \frac {n \log n}{\epsilon^2} $ by an additive factor $ \frac{\log n}{2 \epsilon}$. 
    A direct calculation shows that the value of the density $f(z) = \frac{1}{\sqrt{2\pi \sigma^2}} e^{\frac{(z-\mu)^2}{2\sigma^2}}$ of a normal distribution $N(\mu, \sigma)$ varies by at most a constant multiplicative factor within a standard deviation $\sigma$ from its expected value $\mu$. 
    In other words
    % \begin{equation*}
    %     \begin{cases}
    %       \int_{x}^{\infty}\frac{1}{\sqrt{2\pi \sigma^2}} e^{\frac{(z-\mu)^2}{2\sigma^2}} dz = \frac 12 - \Omega(\frac{x-\mu}{\sigma}) 
    %         & \text{ if }x \in (\mu, \mu + \sigma),\\
    %       \int_{x}^{\infty}\frac{1}{\sqrt{2\pi \sigma^2}} e^{\frac{(z-\mu)^2}{2\sigma^2}} dz = \frac 12 + \Omega(\frac{\mu-x}{\sigma}) 
    %         & \text{ if }x \in (\mu - \sigma, \mu).
    %     \end{cases}
    % \end{equation*}
    
    \[
    \int_{x}^{\infty}\frac{1}{\sqrt{2\pi \sigma^2}} e^{\frac{(z-\mu)^2}{2\sigma^2}} dz = \frac 12 - \Omega(\frac{x-\mu}{\sigma}) 
              \text{ if }x \in (\mu, \mu + \sigma),
    \]
    
    \[
    \int_{x}^{\infty}\frac{1}{\sqrt{2\pi \sigma^2}} e^{\frac{(z-\mu)^2}{2\sigma^2}} dz = \frac 12 + \Omega(\frac{\mu-x}{\sigma}) 
            \text{ if }x \in (\mu - \sigma, \mu).
    \]
    Hence, the probability that, pulling from $\conf 0$, a node chooses $1$ at the end of the first phase, is $\frac 12 - \Omega(\sqrt{\frac{\log n}{n}})$ and the probability that, pulling from $\conf 1$, a node chooses $1$ at the end of the first phase, is $\frac 12 + \Omega(\sqrt{\frac{\log n}{n}})$. 
    
    The proof finally follows from a standard application of the Chernoff bounds (Appendix \ref{apx:cbmf}), to show concentration of probability around the expected value $\frac n2 - \Omega(\sqrt{n{\log n}})$ of nodes supporting $1$ when starting from $\conf 0$, and the expected value $\frac n2 + \Omega(\sqrt{n{\log n}})$ of nodes supporting $1$ when starting from $\conf 1$. 
\end{proof}

\bibliographystyle{plain}
\bibliography{cn}

\begin{thebibliography}{10}

\bibitem{AD12}
M.~Abdullah and M.~Draief.
\newblock Global majority consensus by local majority polling on graphs.
\newblock In {\em 2014 7th International Conference on NETwork Games, COntrol
  and OPtimization (NetGCoop)}, pages 46--53, Oct 2014.

\bibitem{AD}
Mohammed~Amin Abdullah and Moez Draief.
\newblock Global majority consensus by local majority polling on graphs of a
  given degree sequence.
\newblock {\em Discrete Applied Mathematics}, 180:1--10, 2015.

\bibitem{ADHS17}
Abhinav Aggarwal, Varsha Dani, Thomas~P. Hayes, and Jared Saia.
\newblock Distributed computing with channel noise.
\newblock {\em {IACR} Cryptology ePrint Archive}, 2017:710, 2017.

\bibitem{ABLP91}
Noga Alon, Amotz Bar-Noy, Nathan Linial, and David Peleg.
\newblock A lower bound for radio broadcast.
\newblock {\em Journal of Computer and System Sciences}, 43(2):290 -- 298,
  1991.

\bibitem{AAD+06}
Dana Angluin, James Aspnes, Zo{\"{e}} Diamadi, Michael~J. Fischer, and Peralta
  Ren{\'e}.
\newblock {Computation in networks of passively mobile finite-state sensors}.
\newblock {\em Distributed Computing}, 18(4):235--253, 2006.

\bibitem{AAE07}
Dana Angluin, James Aspnes, and David Eisenstat.
\newblock {A Simple Population Protocol for Fast Robust Approximate Majority}.
\newblock {\em Distributed Computing}, 21(2):87--102, 2008.
\newblock (Preliminary version in DISC'07).

\bibitem{AFJ06}
Dana Angluin, Michael~J. Fischer, and Hong Jiang.
\newblock {Stabilizing consensus in mobile networks}.
\newblock In {\em Proc. of Distributed Computing in Sensor Systems (DCOSS'06)},
  volume 4026 of {\em LNCS}, pages 37--50, 2006.

\bibitem{AR09}
James Aspnes and Eric Ruppert.
\newblock {\em An Introduction to Population Protocols}, pages 97--120.
\newblock Springer Berlin Heidelberg, 2009.

\bibitem{AE17}
Chen Avin and Robert Els{\"a}sser.
\newblock Breaking the log n barrier on rumor spreading.
\newblock {\em Distributed Computing}, 2017.

\bibitem{BGI92}
Reuven Bar-Yehuda, Oded Goldreich, and Alon Itai.
\newblock On the time-complexity of broadcast in multi-hop radio networks: An
  exponential gap between determinism and randomization.
\newblock {\em Journal of Computer and System Sciences}, 45(1):104 -- 126,
  1992.

\bibitem{BCNPS15}
Luca Becchetti, Andrea Clementi, Emanuele Natale, Francesco Pasquale, and
  Riccardo Silvestri.
\newblock Plurality consensus in the gossip model.
\newblock In {\em ACM-SIAM SODA'15}, pages 371--390, 2015.

\bibitem{becchetti_plurality_2015}
Luca Becchetti, Andrea Clementi, Emanuele Natale, Francesco Pasquale, and
  Riccardo Silvestri.
\newblock Plurality {Consensus} in the {Gossip} {Model}.
\newblock In {\em Proceedings of the 26th {Annual} {ACM}-{SIAM} {Symposium} on
  {Discrete} {Algorithms} ({SODA})}, {SODA} '15, pages 371--390, 2015.

\bibitem{becchetti2016stabilizing}
Luca Becchetti, Andrea Clementi, Emanuele Natale, Francesco Pasquale, and Luca
  Trevisan.
\newblock Stabilizing consensus with many opinions.
\newblock In {\em Proc. of the 27th Ann. ACM-SIAM Symp. on Discrete
  algorithms}, pages 620--635. SIAM, 2016.

\bibitem{BCNPST14}
Luca Becchetti, Andrea~E.F. Clementi, Emanuele Natale, Francesco Pasquale,
  Riccardo Silvestri, and Luca Trevisan.
\newblock Simple dynamics for plurality consensus.
\newblock In {\em ACM SPAA'14}, pages 247--256, 2014.

\bibitem{PetraAL17}
Petra Berenbrink, Andrea Clementi, Robert Els\"{a}sser, Peter Kling, Frederik
  Mallmann-Trenn, and Emanuele Natale.
\newblock Ignore or comply?: On breaking symmetry in consensus.
\newblock In {\em Proceedings of the ACM Symposium on Principles of Distributed
  Computing}, PODC '17, pages 335--344, 2017.

\bibitem{BFKN18}
Lucas Boczkowski, Emanuele Natale, Ofer Feinerman, and Amos Korman.
\newblock Limits on reliable information flows through stochastic populations.
\newblock {\em PLOS Computational Biology}, 14(6):e1006195, June 2018.

\bibitem{BGPS06}
Stephen Boyd, Arpita Ghosh, Balaji Prabhakar, and Devavrat Shah.
\newblock Randomized gossip algorithms.
\newblock {\em IEEE/ACM Transactions on Networking}, 14:2508--2530, 2006.

\bibitem{cardelli2012cell}
L.~Cardelli and A.~Csik{\'a}sz-Nagy.
\newblock The cell cycle switch computes approximate majority.
\newblock {\em Scientific Reports}, Vol. 2, 2012.

\bibitem{CGH17}
Keren Censor-Hillel, Ran Gelles, and Bernhard Haeupler.
\newblock {Making Asynchronous Distributed Computations Robust to Channel
  Noise}.
\newblock In {\em 9th Innovations in Theoretical Computer Science Conference
  (ITCS 2018)}, volume~94, pages 50:1--50:20, 2018.

\bibitem{C12}
Bernard Chazelle.
\newblock Natural algorithms and influence systems.
\newblock {\em Commun. ACM}, 55(12):101--110, December 2012.

\bibitem{chierichetti_rumour_2010}
F.~Chierichetti, S.~Lattanzi, and A.~Panconesi.
\newblock Rumour spreading and graph conductance.
\newblock In {\em Proceedings of the {Twenty}-{First} {Annual} {ACM}-{SIAM}
  {Symposium} on {Discrete} {Algorithms}}, Proceedings, pages 1657--1663.
  Society for Industrial and Applied Mathematics, January 2010.

\bibitem{chierichetti_almost_2010}
Flavio Chierichetti, Silvio Lattanzi, and Alessandro Panconesi.
\newblock Almost {Tight} {Bounds} for {Rumour} {Spreading} with {Conductance}.
\newblock In {\em Proceedings of the {Forty}-second {ACM} {Symposium} on
  {Theory} of {Computing}}, {STOC} '10, pages 399--408, 2010.

\bibitem{CLP09}
Flavio Chierichetti, Silvio Lattanzi, and Alessandro Panconesi.
\newblock Rumor spreading in social networks.
\newblock {\em Theoretical Computer Science}, 412(24):2602 -- 2610, 2011.
\newblock Selected Papers from 36th International Colloquium on Automata,
  Languages and Programming (ICALP 2009).

\bibitem{CDGLNN08}
Gregory Chockler, Murat Demirbas, Seth Gilbert, Nancy Lynch, Calvin Newport,
  and Tina Nolte.
\newblock Consensus and collision detectors in radio networks.
\newblock {\em Distributed Computing}, 21(1):55--84, June 2008.

\bibitem{clementi_tight_2017}
Andrea E.~F. Clementi, Luciano Gualà, Francesco Pasquale, Giacomo Scornavacca,
  Emanuele Natale, and Mohsen Ghaffari.
\newblock {A Tight Analysis of the Parallel Undecided-State Dynamics with Two
  Colors}.
\newblock In {\em Proc. of the 43rd Int. Symp. on Mathematical Foundations of
  Computer Science (MFCS 2018)}, Leibniz International Proceedings in
  Informatics (LIPIcs), 2018.

\bibitem{CMPS07}
Andrea E.~F. Clementi, Angelo Monti, Francesco Pasquale, and Riccardo
  Silvestri.
\newblock Broadcasting in dynamic radio networks.
\newblock {\em J. Comput. Syst. Sci.}, 75(4):213--230, June 2009.
\newblock Extended abstract in ACM PODC'07.

\bibitem{CMS01}
Andrea E.~F. Clementi, Angelo Monti, and Riccardo Silvestri.
\newblock Selective families, superimposed codes, and broadcasting on unknown
  radio networks.
\newblock In {\em Proceedings of the Twelfth Annual ACM-SIAM Symposium on
  Discrete Algorithms}, SODA '01, pages 709--718, 2001.

\bibitem{CER14}
C.~Cooper, R.~Elsasser, and T.~Radzik.
\newblock The power of two choices in distributed voting.
\newblock In {\em Proceedings of the 41st International Colloquium on Automata,
  Languages, and Programming (ICALP'14)}, volume 8573 of {\em LNCS}, pages
  435--446, 2014.

\bibitem{DGHILSSST87}
A.~Demers, D.~Greene, C.~Hauser, W.~Irish, J.~Larson, S.~Shenker, H.~Sturgis,
  D.~Swinehart, and D.~Terry.
\newblock {Epidemic algorithms for replicated database maintenance}.
\newblock In {\em ACM PODC'87}, 1987.

\bibitem{Dijkstra74}
E.~W. Dijkstra.
\newblock Self-stabilizing systems in spite of distributed control.
\newblock {\em Commun. {ACM}}, 17(11):643--644, 1974.

\bibitem{DDMM16}
Benjamin Doerr, Carola Doerr, Shay Moran, and Shlomo Moran.
\newblock Simple and optimal randomized fault-tolerant rumor spreading.
\newblock {\em Distributed Computing}, 29(2):89--104, Apr 2016.

\bibitem{doerr_quasi_random_2011}
Benjamin Doerr and Mahmoud Fouz.
\newblock Quasi-random rumor spreading: {Reducing} randomness can be costly.
\newblock {\em Information Processing Letters}, 111(5):227--230, February 2011.

\bibitem{doerr_quasirandom_2008}
Benjamin Doerr, Tobias Friedrich, and Thomas Sauerwald.
\newblock Quasirandom {Rumor} {Spreading}.
\newblock {\em ACM Trans. Algorithms}, 11(2):9:1--9:35, October 2014.

\bibitem{DGMSS11}
Benjamin Doerr, Leslie~Ann Goldberg, Lorenz Minder, Thomas Sauerwald, and
  Christian Scheideler.
\newblock Stabilizing consensus with the power of two choices.
\newblock In {\em ACM SPAA'11}, pages 149--158, 2011.

\bibitem{Dolev00}
S.~Dolev.
\newblock {\em Self-Stabilization}.
\newblock The MIT Press, 2000.

\bibitem{Doty14}
David Doty.
\newblock {Timing in chemical reaction networks}.
\newblock In {\em ACM-SIAM SODA'14}, pages 772--784, 2014.

\bibitem{E84}
A.~El~Gamal.
\newblock Open problems.
\newblock In {\em Workshop on specific problems in communication and
  computation}, 1984.

\bibitem{FHK17}
Ofer Feinerman, Bernhard Haeupler, and Amos Korman.
\newblock {Breathe Before Speaking: Efficient Information Dissemination Despite
  Noisy, Limited and Anonymous Communication}.
\newblock {\em Distributed Computing}, 30(5):239--355, 2017.
\newblock Ext. Abs. in ACM PODC'14.

\bibitem{fraigniaud_noisy_2018}
Pierre Fraigniaud and Emanuele Natale.
\newblock Noisy rumor spreading and plurality consensus.
\newblock {\em Distributed Computing}, pages 1--20, June 2018.

\bibitem{HouseHunt}
Nigel~R. Franks, Stephen~C. Pratt, Eamonn~B. Mallon, Nicholas~F. Britton, and
  David~J.T. Sumpter.
\newblock {Information flow, opinion polling and collective intelligence in
  house--hunting social insects}.
\newblock {\em Philosophical Transactions of the Royal Society of London B:
  Biological Sciences}, 357(1427):1567--1583, 2002.

\bibitem{frieze_shortest-path_1985}
A.~M. Frieze and G.~R. Grimmett.
\newblock The shortest-path problem for graphs with random arc-lengths.
\newblock {\em Discrete Applied Mathematics}, 10(1):57--77, January 1985.

\bibitem{FP15}
Emanuele~G. Fusco and Andrzej Pelc.
\newblock Communication {Complexity} of {Consensus} in {Anonymous} {Message}
  {Passing} {Systems}.
\newblock {\em Fundam. Inf.}, 137(3):305--322, July 2015.

\bibitem{G88}
R.~G. Gallager.
\newblock Finding parity in a simple broadcast network.
\newblock {\em IEEE Trans. Inf. Theor.}, 34(2):176--180, September 2006.

\bibitem{GL17}
Mohsen Ghaffari and Johannes Lengler.
\newblock Tight analysis for the 3-majority consensus dynamics.
\newblock {\em CoRR}, abs/1705.05583, 2017.

\bibitem{ghaffari_polylogarithmic_2016}
Mohsen Ghaffari and Merav Parter.
\newblock A {Polylogarithmic} {Gossip} {Algorithm} for {Plurality} {Consensus}.
\newblock In {\em Proceedings of the 36th {ACM} {Symposium} on {Principles} of
  {Distributed} {Computing}}, {PODC} '16, pages 117--126, 2016.

\bibitem{giakkoupis_tight_2011}
George Giakkoupis.
\newblock Tight bounds for rumor spreading in graphs of a given conductance.
\newblock In {\em Symposium on {Theoretical} {Aspects} of {Computer} {Science}
  ({STACS}2011)}, volume~9, pages 57--68, 2011.

\bibitem{giakkoupis_tight_2014}
George Giakkoupis.
\newblock Tight {Bounds} for {Rumor} {Spreading} with {Vertex} {Expansion}.
\newblock In {\em Proceedings of the {Twenty}-{Fifth} {Annual} {ACM}-{SIAM}
  {Symposium} on {Discrete} {Algorithms}}, {SODA} '14, pages 801--815, 2014.

\bibitem{GNW16}
George Giakkoupis, Yasamin Nazari, and Philipp Woelfel.
\newblock How asynchrony affects rumor spreading time.
\newblock In {\em 35th ACM Symposium on Principles of Distributed Computing
  (PODC 2016)}, 2016.

\bibitem{GSa}
George Giakkoupis and Thomas Sauerwald.
\newblock Rumor {Spreading} and {Vertex} {Expansion}.
\newblock In {\em Proceedings of the twenty-third annual ACM-SIAM symposium on
  Discrete Algorithms}, {SODA} '12, pages 1623--1641, 2012.

\bibitem{GKS08}
Navin Goyal, Guy Kindler, and Michael Saks.
\newblock Lower bounds for the noisy broadcast problem.
\newblock {\em SIAM J. Comput.}, 37(6):1806--1841, March 2008.

\bibitem{H15}
Bernhard Haeupler.
\newblock Simple, fast and deterministic gossip and rumor spreading.
\newblock {\em J. ACM}, 62(6):47:1--47:18, December 2015.

\bibitem{HM14}
Bernhard Haeupler and Dahlia Malkhi.
\newblock Optimal gossip with direct addressing.
\newblock In {\em Proceedings of the 2014 ACM Symposium on Principles of
  Distributed Computing}, PODC '14, pages 176--185, 2014.

\bibitem{KSSV00}
R.~Karp, C.~Schindelhauer, S.~Shenker, and B.~Vocking.
\newblock Randomized rumor spreading.
\newblock In {\em IEEE FOCS'00}, pages 565--574, 2000.

\bibitem{KP02}
D.~R. Kowalski and A.~Pelc.
\newblock Deterministic broadcasting time in radio networks of unknown
  topology.
\newblock In {\em The 43rd Annual IEEE Symposium on Foundations of Computer
  Science, 2002. Proceedings.}, pages 63--72, 2002.

\bibitem{KLNOR10}
Fabian Kuhn, Nancy Lynch, Calvin Newport, Rotem Oshman, and Andrea Richa.
\newblock Broadcasting in unreliable radio networks.
\newblock In {\em Proceedings of the 29th ACM SIGACT-SIGOPS Symposium on
  Principles of Distributed Computing}, PODC '10, pages 336--345, 2010.

\bibitem{KNR10}
Fabian Kuhn, Nancy Lynch, and Rotem Oshman.
\newblock Distributed computation in dynamic networks.
\newblock In {\em Proceedings of the Forty-second ACM Symposium on Theory of
  Computing}, STOC '10, pages 513--522, 2010.

\bibitem{KM98}
E.~Kushilevitz and Y.~Mansour.
\newblock An $\omega(d\log (n/d))$ lower bound for broadcast in radio networks.
\newblock {\em SIAM Journal on Computing}, 27(3):702--712, 1998.

\bibitem{lynch_distributed_1996}
Nancy~A Lynch.
\newblock {\em Distributed algorithms}.
\newblock Morgan Kaufmann, 1996.

\bibitem{MacK03}
David J.~C. MacKay.
\newblock {\em Information theory, inference, and learning algorithms}.
\newblock Cambridge University Press, 2003.

\bibitem{MNRS14}
G.~B. Mertzios, S.~E. Nikoletseas, C.~Raptopoulos, and P.~G. Spirakis.
\newblock Determining majority in networks with local interactions and very
  small local memory.
\newblock In {\em Proceedings of the 41st International Colloquium on Automata,
  Languages, and Programming (ICALP'14)}, 2014.

\bibitem{MosselNT14}
Elchanan Mossel, Joe Neeman, and Omer Tamuz.
\newblock Majority dynamics and aggregation of information in social networks.
\newblock {\em Autonomous Agents and Multi-Agent Systems}, 28(3):408--429,
  2014.

\bibitem{navlakha2015distributed}
Saket Navlakha and Ziv Bar-Joseph.
\newblock Distributed information processing in biological and computational
  systems.
\newblock {\em Communications of the ACM}, 58(1):94--102, 2015.

\bibitem{PSL80}
Marshall Pease, Robert Shostak, and Leslie Lamport.
\newblock {Reaching agreement in the presence of faults}.
\newblock {\em Journal of the ACM}, 27(2):228--234, 1980.

\bibitem{PVV09}
Etienne Perron, Dinkar Vasudevan, and Milan Vojnovic.
\newblock {Using Three States for Binary Consensus on Complete Graphs}.
\newblock In {\em IEEE INFOCOM'09}, pages 2527--1535, 2009.

\bibitem{pittel_spreading_1987}
Boris Pittel.
\newblock On {Spreading} a {Rumor}.
\newblock {\em SIAM J. Appl. Math.}, 47(1):213--223, March 1987.

\bibitem{R83}
Michael~O. Rabin.
\newblock {Randomized byzantine generals}.
\newblock In {\em Proc. of the 24th Ann. Symp. on Foundations of Computer
  Science (SFCS)}, pages 403--409. IEEE, 1983.

\bibitem{SW89}
N.~Santoro and P.~Widmayer.
\newblock Time is {Not} a {Healer}.
\newblock In {\em Proceedings of the 6th {Annual} {Symposium} on {Theoretical}
  {Aspects} of {Computer} {Science} on {STACS} 89}, pages 304--313, 1989.

\end{thebibliography}

\newpage
\appendix
\begin{center}
\LARGE{\textbf{Appendix}}
\end{center}
\section{Technical Tools}

\subsection{Symmetry breaking lemma}
\label{lemma4.5}
\begin{lemma}[Lemma 4.5 in \cite{clementi_tight_2017}]
\label{lemma:symmetrygeneric}
    Let $\{X_{t}\}_{t\in \mathbb{N}}$ be a Markov Chain with finite state space $\Omega$ and let 
    $f:\Omega\mapsto[0,n]$ be a function that maps states to integer values. 
    Let $c_3$ be any positive constant
    and let $m = c_3\sqrt{n}\log n$ be a target value.
    Assume the following properties hold:
    \begin{enumerate}
    \item For any positive constant $h$, a positive constant $c_1 < 1$ exists such
    that, for any $x \in \Omega$ with $f(x) < m$, it holds that
    \[
        \Prob{f(X_{t+1}) < h\sqrt{n} \cond X_{t} = x} < c_1\,,
    \]
    
    \item Two positive constants $\epsilon, c_2$ exist such that, for any $x \in
    \Omega$ with $f(x) < m$, it holds that
    \[
        \Prob{f(X_{t+1}) < (1+\epsilon)f(X_{t})\cond X_{t} = x} < e^{-c_2f(x)^2/n}\,.
    \]
    \end{enumerate}
    Then the process reaches a state $x$ such that $f(x) \geq m$ within 
    $\bigO(\log n)$ rounds, w.h.p.
\end{lemma}

\subsection{Berry-Esseen Theorem}
\label{berry-esseen}

\begin{theorem}[Berry-Esseen]
    \label{thm:berreyesseen}
    Let $X_1,\ldots,X_n$ be independent and identically distributed random variables with mean $\mu=0$, variance $\sigma^2 > 0$, and third absolute moment $\rho < \infty$.
    Let $Y_n = \frac{1}{n}\sum_{i=1}^{n}X_i$; 
    let $F_n$ be the cumulative distribution function of $\frac{Y_n\sqrt{n}}{\sigma}$;
    let $\Phi$ the cumulative distribution function of the standard normal distribution.
    Then, there exists a positive constant $C < 0.4748$ such that, for all $x$ and for all $n$, 
    \[
        \vert F_n(x) - \Phi(x) \vert \leq \frac{C \rho}{\sigma^3 \sqrt{n}}.
    \]
\end{theorem}

\subsection{Chernoff Bound multiplicative form} \label{apx:cbmf}

Let $X_1, \dots, X_n$ be independent 0-1 random variables. Let $X = \sum_{i=1}^n X_i$ and $\mu \leq \EEx{}{X} \leq \mu'$. Then, for any $0<\delta<1$ the following Chernoff bounds hold:
\begin{equation}\Prob{X \geq (1+ \delta)\mu} \leq e^{-\mu \delta^2 /3}.\label{eq:cbmfgeq}
\end{equation}
\begin{equation}\Prob{X \leq (1- \delta)\mu'} \leq e^{-\mu' \delta^2 /2}.\label{eq:cbmfleq}
\end{equation}
\subsection{Chernoff Bound additive form} \label{apx:cbaf}

Let $X_1, \dots, X_n$ be independent 0-1 random variables. Let $X = \sum_{i=1}^n X_i$ and $\mu = \EEx{}{X}$. Then the following Chernoff bounds hold:

\noindent for any $0<\lambda<n-\mu$,
\begin{equation}\Prob{X \leq \mu - \lambda} \leq e^{-2\lambda^2/n},\label{eq:cbafleq}
\end{equation}
for any $0<\lambda<\mu$,
\begin{equation}\Prob{X \geq \mu + \lambda} \leq e^{-2\lambda^2/n}.\label{eq:cbafgeq}
\end{equation}

\end{document}